\documentclass[journal,onecolumn,12pt,twoside]{IEEEtran}

\usepackage{graphicx,multicol}
\usepackage{color,soul}
\usepackage{mathtools}
\usepackage{amsmath, amssymb}
\usepackage{amsthm} 
\usepackage{mdframed}
\usepackage{mathrsfs}
\usepackage{cite}
\usepackage{soul}
\usepackage{algorithmic}
\usepackage{array}
\usepackage[font=small]{caption}
\usepackage{subcaption}
\usepackage{cases}
\usepackage{multirow}
\usepackage[draft]{hyperref}
\usepackage[norelsize, linesnumbered, ruled, lined, boxed, commentsnumbered]{algorithm2e}
\usepackage{setspace}
\usepackage[normalem]{ulem}
\usepackage{empheq}
 \usepackage{tikz,lipsum,lmodern}
\usepackage{diagbox}
\usepackage[most]{tcolorbox}
\SetKwInput{KwInput}{Input}                
\SetKwInput{KwOutput}{Output}              

\newtheorem{theorem}{Theorem}
\newtheorem{definition}{Definition}

\newtheorem{proposition}{Proposition}
\newtheorem{lemma}{Lemma}
\newtheorem{rem}{Remark}

\newcommand{\nn}{\nonumber}  
\newcommand{\bieee}{\begin{eqnarray}{rCl}}
\newcommand{\eieee}{\end{eqnarray}}

\makeatletter 
\pretocmd\@bibitem{\color{black}\csname keycolor#1\endcsname}{}{\fail}
\newcommand\citecolor[1]{\@namedef{keycolor#1}{\color{blue}}}
\makeatother

\begin{document}

\title{Constellation Design for Non-Coherent Fast-Forward Relays to Mitigate Full-Duplex Jamming Attacks}

\author{Vivek~Chaudhary and
        Harshan~Jagadeesh
\thanks{V. Chaudhary and H. Jagadeesh are with the Department
of Electrical Engineering, Indian Institute of Technology, Delhi, 110016, India e-mail: (chaudhary03vivek@gmail.com, jharshan@ee.iitd.ac.in).}
\thanks{Parts of this work have been presented in IEEE Globecom, 2021, Madrid, Spain \cite{my_GCOM}.}}

\maketitle

\begin{abstract}
With potential applications to short-packet communication, we address communication of low-latency messages in fast-fading channels under the presence of a reactive jammer. Unlike a traditional jammer, we assume a full-duplex (FD)  jammer capable of detecting pre-existing countermeasures and subsequently changing the target frequency band. To facilitate reliable communication amidst a strong adversary, we propose non-coherent fast-forward full-duplex relaying scheme wherein the victim uses a helper in its vicinity to fast-forward its messages to the base station, in addition to ensuring that the countermeasures are undetected by the FD adversary. Towards designing the constellations for the proposed scheme, we identify that existing non-coherent constellation for fast-fading channels are not applicable owing to the cooperative nature of the fast-forward scheme. As a result, we formulate an optimization problem of designing the non-coherent constellations at the victim and the helper such that the symbol-error-probability at the base station is minimized. We theoretically analyze the optimization problem and propose several strategies to compute near-optimal constellations based on the helper's data-rate and fast-forwarding abilities. We show that the proposed constellations provide near-optimal error performance and help the victim evade jamming. Finally, we also prove the scheme’s efficacy in deceiving the countermeasure detectors at the jammer.
\end{abstract}

\begin{IEEEkeywords}
\centering
Jamming, non-coherent communication, fast-forward relays, full-duplex.
\end{IEEEkeywords}

\IEEEpeerreviewmaketitle

\section{Introduction}
The next generation of wireless networks are pitched to enable new services by providing ultra-reliable and low-latency communication links, such as control of critical infrastructure, autonomous vehicles, and medical procedures. These applications often have mission-critical updates and use short-packet communication with low-rate signalling, e.g. control channel messages (PUCCH) in 5G \cite[Sec.6.3.2]{standard}, and status updates in IoT \cite{SP_DnF}. Since these packets have strict latency constraints, it makes them susceptible to security threats. One popular attack model is the jamming attack, because of which the receiver is unable to decode the packet resulting in deadline violations. Although traditional countermeasures, such as Frequency Hopping (FH) were designed to mitigate jamming attacks, they might not be effective against advanced jamming attacks executed by sophisticated radio devices. Therefore, there is a need to envision new threat models by sophisticated radios and propose strong countermeasures against them to facilitate low-latency communication for the victim.

Among several radio-technologies that have risen in the recent past, the two prominent ones are (i) Full-Duplex (FD) radios with advanced Self-Interference Cancellation (SIC) methods \cite{FD1,FD2,FD3,FD4,FD5,FD6,FD7}, and (ii) Cognitive radios with advanced radio-frequency chains that scan across a wide range of frequency bands. Using these developments, in-band Full-Duplex Cognitive Radio (FDCR) \cite{FDCR1,FDCR2,FDCR3,FDCR4} have been  introduced to scan and transmit in the vacant frequency bands simultaneously, thus improving the network throughput. In line with the motivation of our work, FDCRs have also been studied from an adversarial viewpoint. In particular, \cite{my_PIMRC} and \cite{my_TCCN} introduce an attack model, wherein the adversary, with the help of a \emph{jam-and-measure} FDCR, injects jamming energy on the victim's frequency band and also monitors its energy level after the jamming attack. Owing to the use of jam-and-measure FDCRs, \cite{my_PIMRC} and \cite{my_TCCN} also point out that the state-of-art countermeasures, like FH are ineffective, since the attacker can detect that the victim has vacated the jammed frequency band. As a consequence, they also propose several countermeasures wherein the victim node seeks assistance from a Fast-Forward FD (FFFD) \cite{FD8} relay to instantaneously forward its messages to the base station without getting detected by the FDCR. With the use of fast-forward relays, the countermeasures capture the best-case benefits in terms of facilitating low-latency communication for the victim node. 

Inspired by \cite{my_PIMRC} and \cite{my_TCCN}, we identify that FDCRs can also scan multiple frequencies while executing a \emph{jam-and-measure} attack on the victim's frequency. Subsequently, this can allow the adversary to compute a correlation measure between the symbols on the victim's frequency and other frequencies thereby detecting repetition coding across frequencies, such as the FFFD based countermeasures in \cite{my_PIMRC} and \cite{my_TCCN}. Thus, new countermeasures must be designed to mitigate adversaries which can scan multiple frequencies, in addition to monitoring the energy level on the jammed frequency band. We also point out that the modulation techniques designed as part of the countermeasures depend on the wireless environment. For instance, in slow-fading channels, coherent modulation based countermeasures must be designed by allowing the receiver to learn the Channel State Information (CSI) through pilots. However, acquiring CSI using pilots is difficult when channel conditions vary rapidly over time. As a result, non-coherent modulation based countermeasures must be designed when jam-and-measure attacks are executed in fast-fading channels, thereby allowing the receiver to decode the information symbols without instantaneous CSI. From the above discussion, we identify that the countermeasures proposed in \cite{my_PIMRC} and \cite{my_TCCN} are not applicable for fast-fading channels, thereby opening up new problem statements in designing non-coherent modulation based countermeasures. 

\subsection{Contribution}
In this work, we design non-coherent modulation based countermeasures to mitigate jamming attacks by FDCRs. Amongst various non-coherent modulation techniques, we use energy detection based Amplitude Shift Keying (ASK) due to its higher spectral efficiency. Towards this end, we summarize the contribution of this work as follows: 

\begin{enumerate}
\item We envisage an attack model wherein the adversary uses an FDCR to jam a victim that has low-latency symbols to communicate with the base station. The salient feature of the adversary is that it can scan multiple frequencies in the network while executing a jamming attack on the victim's frequency. In particular, the adversary uses an Energy Detector (ED) and a Correlation Detector (CD) to detect the state-of-art countermeasures. (See Sec.~\ref{sec:systemmodel})

\item As a countermeasure against the proposed threat, we propose a Non-Coherent FFFD (NC-FFFD) relaying scheme, wherein an FFFD helper assists the victim by instantaneously fast-forwarding victim's message along with its message to the base station. The proposed NC-FFFD scheme also uses a Gold-sequence based scrambler to cooperatively pour energy on the victim's frequency in order to evade detection by ED and CD. With On-Off Keying (OOK) at the victim and $M-$ary ASK at the helper, we propose an approximate joint maximum a posteriori  decoder to compute the closed-form expression of symbol error probability for the NC-FFFD scheme. We then formulate an optimization problem of minimizing the SEP over the victim's and the helper's energy levels, subject to a modified average energy constraint at the helper. Subsequently, we solve the optimization problem for $M=2$ and then generalise it for $M>2$. (See Sec.~\ref{sec:NCFFFD},~\ref{sec:optimization})

\item We also consider the case when fast-forwarding at the helper is not instantaneous, i.e., imperfect fast-forwarding. Here, we propose Delay Tolerant NC-FFFD (DT NC-FFFD) scheme, where we solve the optimization problem similar to $M\geq 2$ by upper bounding the energy contributed by the victim by a small number. We show that the error performance of DT NC-FFFD scheme is independent of the delays introduced due to imperfect fast-forwarding. For all the cases, we provide strong analytical results and based on these results, we provide a family of algorithms to obtain near-optimal solutions to the optimization problem. (See Sec.~\ref{sec:DT_NC-FFFD})

\item Finally, through various analytical and simulation results, we show that despite having robust detectors, the adversary cannot detect the proposed mitigating scheme with high probability. (See Sec.~\ref{sec:Covert})
\end{enumerate}

\begin{figure}
\vspace{-0.25in}
\centering
\includegraphics[scale = 0.23]{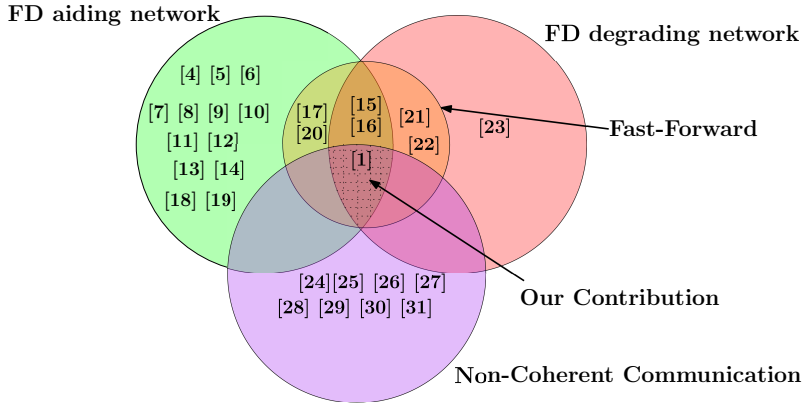}
\caption{\label{fig:venn} Novelty of our work w.r.t. existing contributions.}
\end{figure}

\subsection{Related Work and Novelty}
FD radios have found their applications in aiding   \cite{my_PIMRC,my_TCCN,FD8}, \cite{Aid_FD_1,Aid_FD_2,Aid_FD_3} as well as degrading  \cite{my_PIMRC}, \cite{my_TCCN}, \cite{Foe_FD_1,Foe_FD_2,Foe_FD_3} a network's performance. Along the lines of \cite{my_PIMRC} and \cite{my_TCCN}, this work also uses FD radios at both the adversary and the helper node. However, in contrast, the threat model of this work is stronger than the one in \cite{my_PIMRC} and \cite{my_TCCN} as it can scan multiple frequencies to measure correlation between symbols on different frequencies. Furthermore, the FD radio at the helper in this work implements non-coherent modulation as against coherent modulation in \cite{my_PIMRC} and \cite{my_TCCN}. Our work can be viewed as a constellation design problem for a NC-FFFD strategy. In literature, \cite{ranjan,NC-p2p1,Goldsmith2,NC-p2p2,NC_Gao,new_ref} address the problem of constellation design for point-to-point Single-Input Multiple-Output (SIMO) non-coherent systems. Further, \cite{Goldsmith1}, \cite{Joint-MAC} study the constellation design for non-coherent Multiple Access Channel (MAC). However, due to distributed setting, our work cannot be viewed as a direct extension of \cite{ranjan,NC-p2p1,Goldsmith2,NC-p2p2,NC_Gao, new_ref,Goldsmith1,Joint-MAC}, as pointed in Fig.~\ref{fig:venn}. Some preliminary results on the NC-FFFD scheme have been presented by us in \cite{my_GCOM}, where we compute the optimal energy levels at the victim and the helper for $M=2$. In addition, the results of this work are generalisable for $M\geq 2$. Further, we provide solutions for imperfect fast-forwarding at the helper and also present an extensive analysis on the covertness of the proposed schemes.

\section{System Model}
\label{sec:systemmodel}

We consider a \emph{crowded} network wherein multiple nodes communicate with a base station on orthogonal frequencies. In the context of this work, crowded network implies that all the nodes use orthogonal frequency bands to communicate with the base station such that the number of frequency bands is equal to the number of nodes in the network. Fig.~\ref{fig:NCFFFD}a captures one simple instantiation of the network where two nearby nodes, Alice and Charlie communicate with a multi-antenna base station, Bob. The uplink frequencies of Alice and Charlie are centred around $f_{AB}$ and $f_{CB}$, respectively. Alice is a single-antenna transmitter that has low-rate and low-latency messages to communicate with Bob. In contrast, Charlie, which is a Fast-Forward Full-Duplex (FFFD) node with  $N_{C}$ receive-antennas and a single transmit-antenna, has arbitrary data-rate messages to communicate with no latency constraints. Here, fast-forwarding \cite{FD8} refers to Charlie's capability to instantaneously manipulate the received symbols on its uplink frequency and then multiplex them along with its information symbols to the base station. The mobility conditions of the network are such that the wireless channels from Alice to Bob, and from Charlie to Bob experience fast-fading with a coherence-interval of one symbol duration. Therefore, both Alice and Charlie use non-coherent Amplitude Shift Keying (ASK) for uplink communication. In particular, since Alice has low data-rate messages, she uses the On-Off Keying (OOK) scheme. On the other hand, since Charlie transmits at arbitrary data-rates, he uses an $M$-ary ASK scheme, for some $M = 2^{m}$, with $m \geq 1$. 
\begin{figure}[t]
\vspace{-0.25in}
\centering
\includegraphics[width = 0.73\textwidth, height = 0.3\textwidth]{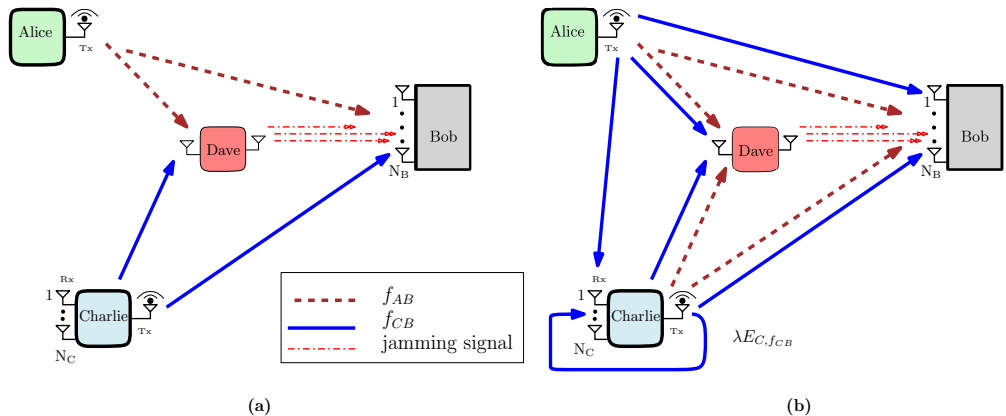}
\caption{\label{fig:NCFFFD}(a) A network model consisting legitimate nodes Alice and Charlie communicating with Bob, on $f_{AB}$, and $f_{CB}$, respectively. Dave is the FD adversary, jamming $f_{AB}$. He also measures the energy level on $f_{AB}$ and computes the correlation between the symbols on $f_{AB}$ and $f_{CB}$. (b) System model of NC-FFFD relaying scheme.}
\end{figure}

Within the same network, we also consider an adversary, named Dave, who is a cognitive jamming adversary equipped with an FD radio that constantly jams $f_{AB}$ and also monitors it to detect any countermeasures. We assume that Dave can learn Alice's frequency band by listening to the radio resource assignment information broadcast from the base station \cite{PRB}. To monitor $f_{AB}$ for any possible countermeasures, Dave uses an Energy Detector (ED), which measures the average energy level on $f_{AB}$. Furthermore, assuming that Dave does not have the knowledge of helper's frequency band, he uses a Correlation Detector (CD) that estimates the correlation between the symbols on $f_{AB}$ and all other frequencies in the network. To envision a practical adversarial model, we assume that Dave's FD radio experiences residual SI. From the above threat model, we note that Alice cannot use frequency hopping to evade the jamming attack due to two reasons: (i) the absence of vacant frequency bands in the uplink, and moreover, (ii) an ED at Dave restricts Alice to completely switch her operating frequency. This is because, if Alice switches her operating frequency, Dave measures a significant dip in the energy level of $f_{AB}$, thereby detecting a countermeasure. Other than frequency hopping, if Alice implements repetition coding using frequency-diversity techniques, where she replicates her messages on $f_{AB}$ and another frequency (say $f_{CB}$), simultaneously, then the CD at Dave detects a high correlation between the symbols on both the frequencies. Subsequently, a detection by either ED or CD compels Dave to jam $f_{CB}$ thereby degrading the network's performance. Therefore, Alice must use a countermeasure that helps her to communicate reliably with Bob while deceiving ED and CD at Dave.

For ease of understanding, in Table~\ref{tab:notations}, we have provided the notations that appear in the rest of the paper. In the next section, we present a communication setting wherein Alice seeks assistance from Charlie to evade the jamming attack whilst deceiving the ED and the CD at Dave.

\begin{table}[!htb]
    \caption{\label{tab:notations}FREQUENTLY OCCURRING NOTATIONS}
    \begin{minipage}[t]{.5\linewidth}
      \centering
      \scalebox{0.8}{
         \begin{tabular}{ | m{2em} m{8cm} | } 
  \hline
  $N_{C}$  & Receive antennas at Charlie  
  \\ 
  $N_{B}$ & Receive antennas at Bob  
  \\ 
  $M$ & Charlie's constellation size 
  \\ 
  $E_{A}$ & Alice's OOK symbol
  \\
  $E_{C}$ & Charlie's multiplexed symbol
  \\
  $\mathbf{r}_{C}$ & $N_{C}\times 1$ receive vector at Charlie
  \\
  $\Omega_{i}$ & Energy received at Charlie corresponding to Alice's $i^{th}$ symbol
  \\
  $\mathbf{r}_{B}$ & $N_{B}\times 1$ receive vector at Bob
  \\
  $\mathbf{0}_{N_{C}}$ & $N_{C}\times 1$ vector of zeros
  \\
   $\mathbf{I}_{N_{C}}$ & $N_{C}\times N_{C}$ Identity matrix
  \\
  $S_{\ell}$ & Sum energy received at Bob on $f_{CB}$
  \\ \hline
   \end{tabular}
   }
    \end{minipage}%
    \begin{minipage}[t]{.5\linewidth}
      \centering
      \scalebox{0.8}{
        \begin{tabular}{ | m{2em} m{8cm} | } 
  \hline
   $\nu$ & Detection threshold at Charlie
   \\
  $\rho_{\ell\ell^{*}}$ & Detection threshold between $S_{\ell}$ and $S_{\ell^{*}}$
  \\
  $\lambda$ & Residual self-interference
  \\
  $\alpha$ & Energy splitting factor
  \\
  $L$ & Number of symbols received at Dave
  \\
  $E_{C,f_{AB}}$ & Avg. transmit energy of Charlie on $f_{CB}$
  \\
  $E_{D,f_{AB}}$ & Avg. receive energy of Dave on $f_{AB}$
  \\
  $r_{D}(l)$ & $l^{th}$ symbol received at Dave
  \\
  $\mathbf{P}_{FA}$ & Probability of false-alarm at Dave before implementing the countermeasure.
  \\
   $\mathbf{P}_{D}$ & Probability of detection at Dave after implementing the countermeasure.
   \\ \hline
   \end{tabular}
   }
    \end{minipage}
    
\end{table}

\section{Non-Coherent FastForward Full-Duplex Relaying Scheme (NC-FFFD)}
\label{sec:NCFFFD}

In order to help Alice evade the jamming attack, we propose a Non-Coherent Fast-Forward Full-Duplex (NC-FFFD) relaying scheme, described as follows: Bob directs Alice to broadcast her OOK symbols on $f_{CB}$ with $(1-\alpha)$ fraction of her energy, where $\alpha\in (0,1)$ is a design parameter. Since Charlie also has symbols to communicate to Bob, he uses his in-band FD radio to receive Alice's symbols on $f_{CB}$, decodes them, multiplexes them to his symbols, and then \emph{fast-forwards} them on $f_{CB}$, such that the average energy of the multiplexed symbols is $(1+\alpha)/2$ fraction of his original average energy. As a result, Bob observes a MAC on $f_{CB}$, and attempts to decode Alice's and Charlie's symbols jointly. To deceive the ED at Dave, the average energy level on $f_{AB}$ needs to be the same as before implementing the countermeasure. Therefore, Alice and Charlie use a Gold sequence-based scrambler as a pre-shared key to cooperatively transmit dummy OOK symbols on $f_{AB}$ by using residual $\alpha/2$ and $(1-\alpha)/2$ fractions of their average energies, respectively. Note that the use of dummy OOK symbols on $f_{AB}$ assists in deceiving the CD at Dave. In the next section, we discuss the signal model of the NC-FFFD scheme on $f_{CB}$ so as to focus on reliable communication of Alice's symbols with the help of Charlie.

\subsection{Signal Model} 
\label{ssec:signalmodel}
Before implementing the NC-FFFD relaying scheme, Alice transmits her OOK symbols with energy $E_{A} \in \{0, 1\}$, such that $E_{A}=0$ and $E_{A}=1$ correspond to symbols $i=0$ and $i=1$, respectively. Similarly, Charlie transmits his symbols using an $M-$ary ASK scheme with average energy $1$. When implementing the NC-FFFD relaying scheme, as illustrated in Fig.~\ref{fig:NCFFFD}b, Alice transmits her OOK symbols with energy $(1-\alpha)E_{A}$, for some $\alpha \in (0, 1)$ on $f_{CB}$. With this modification, the average transmit energy of Alice on $f_{CB}$, denoted by $\mathrm{E}_{A,f_{CB}}$, is $\mathrm{E}_{A,f_{CB}} = (1-\alpha)/2$. Since Charlie is an in-band FD radio, the received baseband vector at Charlie on $f_{CB}$ is,
\bieee
\mathbf{r}_{C} = \mathbf{h}_{AC}\sqrt{(1-\alpha)E_{A}} + \mathbf{h}_{CC} + \mathbf{n}_{C},\label{eq:rc}
\eieee
\noindent where $\mathbf{h}_{AC}\sim{\cal CN}\left(\mathbf{0}_{N_{C}},\sigma_{AC}^{2}\mathbf{I}_{N_{C}}\right)$ is $N_{C}\times 1$ channel vector. Further, $\mathbf{h}_{CC}\sim{\cal CN}\left(\mathbf{0}_{N_{C}},\lambda\mathrm{E}_{C,f_{CB}}\mathbf{I}_{N_{C}}\right)$ denotes the SI channel of the FD radio at Charlie \cite{my_TCCN}. Finally, $\mathbf{n}_{C}\sim{\cal CN}\left(\mathbf{0}_{N_{C}},N_{o}\mathbf{I}_{N_{C}}\right)$ is the $N_{C}\times 1$ Additive White Gaussian Noise (AWGN) vector.

As a salient feature of the NC-FFFD scheme, Charlie uses $\mathbf{r}_{C}$ to instantaneously decode Alice's information symbol, and then transmits an energy level $E_{C}$ on $f_{CB}$, which is a function of Alice's decoded symbol and its information symbol. If $\hat{i}_{C}$ and $j\in\{1,\cdots,M\}$ denote Alice's decoded symbol and Charlie's information symbol, respectively, then the energy level, $E_{C}$ is given by
\begin{equation}
E_{C} =
\begin{cases}
\epsilon_{j} & \text{if } \hat{i}_{C}=0,
\\
\eta_{j} & \text{if } \hat{i}_{C}=1.
\end{cases}
\label{eq:multiplexing_symbol}
\end{equation}
Here $\{\epsilon_{j}, \eta_{j} ~|~ j = 1, \cdots, M\}$, which represent the set of energy levels corresponding to different combinations of $\hat{i}_{C}$ and $j$, are the parameters under design consideration. Note that Charlie transmits $M$ energy levels corresponding to each value of $\hat{i}_{C}$. Towards designing $\{\epsilon_{j}, \eta_{j}\}$, the energy levels are such that, $0\leq\epsilon_{1}<\cdots<\epsilon_{M}$,  $0\leq\eta_{1}<\cdots<\eta_{M}$, and $\epsilon_{j} < \eta_{j}$, if $j$ is odd and $\epsilon_{j} > \eta_{j}$, if $j$ is even.

Given that Alice contributes an average energy of $(1-\alpha)/2$ on $f_{CB}$, Charlie is constrained to transmit his multiplexed symbols with an average energy of $(1+\alpha)/2$ so that the average energy on $f_{CB}$ continues to be unity. Thus, the average energy constraint on Charlie, denoted by $\mathrm{E}_{C,f_{CB}}$ is,
\bieee
\mathrm{E}_{C,f_{CB}} = \dfrac{1}{2M}\sum_{j=1}^{M}(\epsilon_{j}+\eta_{j}) &=& \dfrac{1+\alpha}{2}.\label{eq:new_constaint}
\eieee
Finally, upon transmission of the energy level $E_{C}$ from Charlie, Bob observes a multiple access channel on $f_{CB}$ from Alice and Charlie. Thus, the $N_{B}\times 1$ receive vector at Bob is,
\bieee
\mathbf{r}_{B} = \mathbf{h}_{AB}\sqrt{(1-\alpha)E_{A}} + \mathbf{h}_{CB}\sqrt{E_{C}} + \mathbf{n}_{B},\label{eq:rb}
\eieee
\noindent where $\mathbf{h}_{AB}\sim{\cal CN}\left(\mathbf{0}_{N_{B}},\sigma_{AB}^{2}\mathbf{I}_{N_{B}}\right)$, $\mathbf{h}_{CB}\sim{\cal CN}\left(\mathbf{0}_{N_{B}},\sigma_{CB}^{2}\mathbf{I}_{N_{B}}\right)$, and $\mathbf{n}_{B}\sim{\cal CN}\left(\mathbf{0}_{N_{B}},N_{o}\mathbf{I}_{N_{B}}\right)$ are the Alice-to-Bob link, Charlie-to-Bob link and the AWGN vector at Bob. We assume that all the channel realizations and noise samples are statistically independent. We also assume that only the channel statistics and not the instantaneous realizations of $\mathbf{h}_{AB}$ and $\mathbf{h}_{CB}$  are known to Bob. Similarly, only the channel statistics and not the instantaneous realizations of $\mathbf{h}_{AC}$ are known to Charlie. Further, due to the proximity of Alice and Charlie, we assume $\sigma_{AC}^{2}>\sigma_{AB}^{2}$ to capture higher Signal-to-Noise Ratio (SNR) in the Alice-to-Charlie link compared to Charlie-to-Bob link. Henceforth, throughout the paper, various noise variance at Charlie and Bob are given by $\text{SNR} = N_{o}^{-1}$ and $\sigma_{AB}^{2} = \sigma_{CB}^{2} = 1$. 

Given that Alice-to-Bob and Charlie-to-Bob channels are non-coherent, Bob must use $\mathbf{r}_{B}$ in \eqref{eq:rb} to jointly decode the information symbols of both Alice and Charlie. Towards that direction, in the next section, we study the distribution on $\mathbf{r}_{B}$ conditioned on their information symbols.
  
\subsection{The Complementary Energy Levels and Distribution of $\mathbf{r}_{B}$}
\label{ssec:com_energy}

Based on the MAC in \eqref{eq:rb}, $\mathbf{r}_{B}$ is sampled from an underlying multi-dimensional Gaussian distribution whose parameters depend on $i$, $j$, and $\hat{i}_{C}$. If $e$ denotes the error event at Charlie, then, $e=0$, if $i=\hat{i}_{C}$ and $e=1$, if $i\neq \hat{i}_{C}$.
Recall that for a given $j$, Charlie transmits $\epsilon_{j}$ or $\eta_{j}$ corresponding to $\hat{i}_{C}=0$ and $\hat{i}_{C}=1$, respectively. Therefore, Bob receives $\mathbf{r}_{B}$ sampled from two different sets with $2M$ multi-dimensional Gaussian distributions corresponding to $e=0$ and $e=1$. For example, assume that Alice transmits symbol $i=1$, and it gets decoded as $\hat{i}_{C}=0$ at Charlie. According to \eqref{eq:multiplexing_symbol}, Charlie transmits the energy level $\epsilon_{j}$, and as a result, each component of $\mathbf{r}_{B}$ is sampled from a circularly symmetric complex Gaussian distribution with mean zero and variance $1-\alpha+\epsilon_{j}+N_{o}$. On the other hand, if Charlie had decoded the symbol correctly, each component of $\mathbf{r}_{B}$ would be sampled from a circularly symmetric complex Gaussian distribution with mean zero and variance $1-\alpha + \eta_{j}+N_{o}$. To obtain these variance values, we have used the fact that $\mathbf{h}_{AB}\sim{\cal CN}\left(\mathbf{0}_{N_{B}},\mathbf{I}_{N_{B}}\right)$, $\mathbf{h}_{CB}\sim{\cal CN}\left(\mathbf{0}_{N_{B}},\mathbf{I}_{N_{B}}\right)$, and $\mathbf{n}_{B}\sim{\cal CN}\left(\mathbf{0}_{N_{B}},N_{o}\mathbf{I}_{N_{B}}\right)$. Overall, using \eqref{eq:rb}, the distribution of $\mathbf{r}_{B}$ is given as,
\bieee
\mathbf{r}_{B}\sim
\begin{cases}
{\cal CN}\left(\mathbf{0}_{N_{B}},(\epsilon_{j} + N_{o})\mathbf{I}_{N_{B}}\right) & \text{if } i=0,e=0,
\\
{\cal CN}\left(\mathbf{0}_{N_{B}},(\eta_{j} + N_{o})\mathbf{I}_{N_{B}}\right) & \text{if } i=0,e=1,
\\
{\cal CN}\left(\mathbf{0}_{N_{B}},(1-\alpha+\eta_{j} + N_{o})\mathbf{I}_{N_{B}}\right) & \text{if } i=1,e=0,
\\
{\cal CN}\left(\mathbf{0}_{N_{B}},(1-\alpha+\epsilon_{j} + N_{o})\mathbf{I}_{N_{B}}\right) & \text{if } i=1,e=1,
\end{cases}
\label{eq:rb_distribution1}
\eieee
\noindent where we have substituted $E_{A}\!=\!0$ and $E_{A}\!=\!1$, for $i=0$ and $i=1$, respectively, and $\sigma_{AB}^{2}=\sigma_{CB}^{2}=1$ in \eqref{eq:rb}. From \eqref{eq:rb_distribution1}, it is clear that the sum of the energy levels transmitted by Alice and Charlie characterizes all the possible distributions from which $\mathbf{r}_{B}$ is sampled.
 
We now define an index $\ell$ that is a one-to-one function of the  transmit pair $(i,j)$, such that 
\bieee
\ell = \frac{1}{2}\left[(-1)^{ij}\left(4j(1-i) + 4i(-1)^{j}+(-1)^{j+i}-1\right)\right].\label{eq:def_l}
\eieee
\noindent Since $(i,j)\in\{0,1\}\times\{1,\cdots, M\}$, we have $\ell\in\{1,\cdots, 2M\}$. We also define two sets of energy levels, denoted by $\mathcal{S}=\{S_{\ell}~\vert~\ell = 1,\cdots,2M\}$ and $\mathcal{\overline{S}}=\{\overline{S}_{\ell}~\vert~ \ell=1,\cdots,2M\}$ that correspond to the sum of energy levels jointly contributed by Alice and Charlie, and the AWGN at Bob when $e=0$ and $e=1$, respectively. In particular, the $\ell^{th}$ element of $\mathcal{S}$ and $\mathcal{\overline{S}}$ are given by
\bieee
S_{\ell} \triangleq  \left(1-\alpha+\eta_{j}\right)i+\epsilon_{j}(1-i)+N_{o} \text{ and } \overline{S}_{\ell} \triangleq  \left(1-\alpha+\epsilon_{j}\right)i+\eta_{j}(1-i)+N_{o}.\label{eq:map2}
\eieee
\noindent Since $\mathcal{\overline{S}}$, corresponds to the sum of energy levels when $e=1$, we refer to $\mathcal{\overline{S}}$ as the set of complementary energy levels. Note that there is one-to-one correspondence between the elements of $\mathcal{S}$ and $\mathcal{\overline{S}}$, and the distributions in \eqref{eq:rb_distribution1} corresponding to $e=0$ and $e=1$, respectively. Also, note that $\mathcal{S}$ is such that $S_{1}<S_{2}<\cdots<S_{2M-1}<S_{2M}$. To exemplify the sum of energy levels that characterises $\mathbf{r}_{B}$ at Bob, in Fig.~\ref{fig:consexample}, we present the elements of $\mathcal{S}$ and $\mathcal{\overline{S}}$ for $M=2,4$. 
\begin{figure}[t]
\vspace{-0.25in}
\centering
\includegraphics[scale = 0.35]{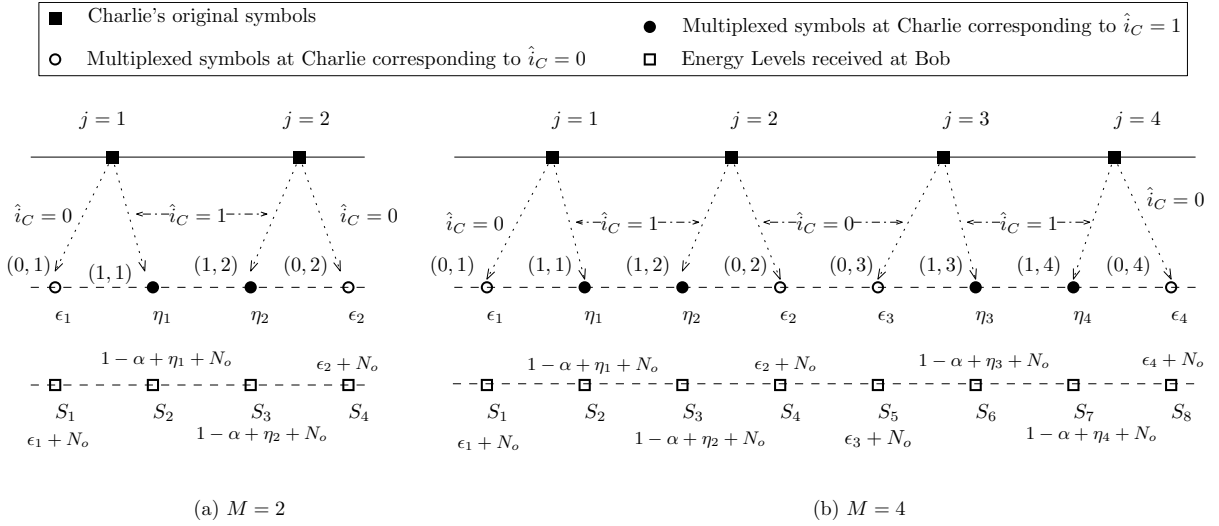}
\caption{\label{fig:consexample}Illustration of multiplexing at Charlie and corresponding energy levels received at Bob.}
\end{figure}

\subsection{Joint Maximum A Posteriori (JMAP) decoder for NC-FFFD Relaying Scheme}
\label{ssec:JMAP}

Due to the decode-multiplex-and-forward nature of the NC-FFFD scheme, we first compute the error-rates introduced by Charlie while decoding Alice's symbols, and then compute the joint error-rates at Bob. Since Alice-to-Charlie link is non-coherent, Charlie uses energy detection to decode Alice's symbols. If $f\left(\mathbf{r}_{C}\vert i\right)$ is the PDF of $\mathbf{r}_{C}$ conditioned on the Alice's symbol, $i$, then the Maximum Likelihood (ML) decision rule for detection is
\bieee
\hat{i}_{C} = \arg\underset{i\in\{0,1\}}{\max\ }\ln\left\{f\left(\mathbf{r}_{C}\vert i\right)\right\} = \arg\underset{i\in\{0,1\}}{\min\ }N_{C}\ln(\pi\Omega_{i}) + \dfrac{\mathbf{r}_{C}^{H}\mathbf{r}_{C}}{\Omega_{i}},\label{eq:rule_rc}
\eieee
\noindent where $(\mathbf{r}_{C}\vert i)\sim{\cal CN}\left(\mathbf{0}_{N_{C}}, \Omega_{i}\mathbf{I}_{N_{C}}\right)$, such that $\Omega_{0} =  \left(\lambda\frac{(1+\alpha)}{2}+N_{o}\right)$ and $\Omega_{1} =  \left(\sigma_{AC}^{2}(1-\alpha) + \lambda\frac{(1+\alpha)}{2}+N_{o}\right)$ are the variance of the received symbol, when $i=0$ and $i=1$, respectively. Here, $(\cdot)^{H}$ denotes the Hermitian operator. Using the first principles, the energy detection threshold at Charlie, denoted by $\nu$, is given as, $N_{C}\frac{\Omega_{0}\Omega_{1}}{\Omega_{0}-\Omega_{1}}\ln\left(\frac{\Omega_{0}}{\Omega_{1}}\right)$. Using $\nu$, it is straightforward to prove the next theorem that presents the probability of error at Charlie in decoding Alice's symbols.

\begin{theorem}
\label{th:P01P10}
If $P_{ik}$ denotes the probability of decoding symbol $i$ as symbol $k$, for $i,k=0,1$, then  $P_{01} = \frac{\Gamma\left(N_{C}, \frac{\nu}{\Omega_{0}}\right)}{\Gamma(N_{C})}$ and $P_{10} = \frac{\gamma\left(N_{C}, \frac{\nu}{\Omega_{1}}\right)}{\Gamma(N_{C})}$, where $\gamma(\cdot,\cdot)$, $\Gamma(\cdot,\cdot)$, and $\Gamma(\cdot)$ are incomplete lower, incomplete upper, and complete Gamma functions, respectively. 
\end{theorem}

\begin{lemma}
\label{lm:P10P01_alpha}
The terms $P_{01}$ and $P_{10}$ are increasing functions of $\alpha$ for a given SNR, $N_{C}$, and $\lambda$. 
\end{lemma}
\begin{proof}
Consider the expression of $P_{10}$ as given in Theorem~\ref{th:P01P10}. The ratio, $\nu/\Omega_{1}$ can be rewritten as, $\frac{\nu}{\Omega_{1}} = N_{C}\frac{\ln\left(1+\theta\right)}{\theta}$, where $\theta =\frac{ \left(\Omega_{1}-\Omega_{0}\right)}{\Omega_{0}}$. Differentiating $\theta$ w.r.t. $\alpha$, we get, $-\frac{N_{o}\sigma_{AC}^{2}}{\left(N_{o} + \lambda\frac{1+\alpha}{2}\right)^{2}}$. Since $\frac{d\theta}{d\alpha}<0$, as $\alpha$ increases $\theta$ decreases. Further, when $\theta$ decreases, $N_{C}\frac{\ln(1+\theta)}{\theta}$ increases. Therefore, $\frac{\nu}{\Omega_{1}}$ is an increasing function of $\alpha$. Finally, since  $\gamma\left(N_{C}, \frac{\nu}{\Omega_{1}}\right)$ is an increasing function of $\frac{\nu}{\Omega_{1}}$, $P_{10}$ is an increasing function of $\alpha$. Using similar argument, we can prove that $P_{01}$ is also an increasing function of $\alpha$.
\end{proof}
Along the similar lines of Lemma~\ref{lm:P10P01_alpha}, the following lemma is also straightforward to prove.
\begin{lemma}
\label{lm:P10P01_nc}
The terms $P_{01}$ and $P_{10}$ are decreasing functions of $N_{C}$ for a fixed SNR, $\alpha$, and $\lambda$.
\end{lemma}

Using $P_{01}$ and $P_{10}$ at Charlie, we study the performance of non-coherent decoder at Bob. With $i \in \{0, 1\}$ and $j \in \{1, 2, \ldots, M\}$ denoting Alice's and Charlie's information symbols, respectively, we define a transmit pair as the two-tuple $(i,j)$. Based on $\mathbf{r}_{B}$ in \eqref{eq:rb}, the JMAP decoder at Bob is
\bieee
\hat{i},\hat{j} = \arg\underset{i\in\{0,1\},j\in\{1,\cdots,M\}}{\max}g\left(\mathbf{r}_{B}\vert (i,j)\right),\label{eq:JMAP}
\eieee
\noindent where $g\left(\mathbf{r}_{B}\vert (i,j)\right)$ is the PDF of $\mathbf{r}_{B}$, conditioned on $i$ and $j$. However, note that due to errors introduced by Charlie in decoding Alice's symbols, $g(\cdot)$ is a Gaussian mixture for each realization of $i$. The conditional PDF of $g\left(\mathbf{r}_{B}\vert (i,j)\right)$ for $i = 0,1$ is,
\bieee
g\left(\mathbf{r}_{B}\vert (i,j)\right) &=& P_{ii}g\left(\mathbf{r}_{B}\vert (i,j), e=0\right)+ 
 P_{i\overline{i}}g\left(\mathbf{r}_{B}\vert (i,j), e=1\right),\label{eq:JMAP_GM1}
\eieee 
\noindent where $g\left(\mathbf{r}_{B}\vert (i,j), e=0\right)$ and $g\left(\mathbf{r}_{B}\vert (i,j), e=1 \right)$ are the PDFs given in \eqref{eq:rb_distribution1} and $\overline{i}$ is the complement of $i$. Since solving the error performance of the JMAP decoder using the Gaussian mixtures in \eqref{eq:JMAP_GM1} is non-trivial, we approximate the JMAP decoder by only considering the dominant terms in the summation of \eqref{eq:JMAP_GM1} \cite{my_TCCN} to obtain
\bieee
\hat{i},\hat{j} = \arg\underset{i\in\{0,1\},j\in\{1,\cdots,M\}}{\max\ }\tilde{g}\left(\mathbf{r}_{B}\vert (i,j), e=0\right),\label{eq:JD}
\eieee
\noindent where $\tilde{g}\left(\mathbf{r}_{B}\vert (i,j),e=0\right)$ is the first term on the RHS of \eqref{eq:JMAP_GM1}. Henceforth, we refer to the above decoder as the Joint Dominant (JD) decoder. To showcase the accuracy of the approximation in \eqref{eq:JD}, we tabulate the error-rates for arbitrary energy levels and $\alpha$ for JMAP and JD decoders in Table~\ref{tab:approximation_JMAP_JD}. We compute the relative-error between error-rates of JMAP and JD decoder as, $\left\vert\frac{{P\textsubscript{JMAP}}-{P\textsubscript{JD}}}{{P\textsubscript{JMAP}}}\right\vert$ and show that the maximum relative error is within $5.55\%$. Therefore, in the next section, we discuss the error analysis using JD decoder.
\begin{table}[!h]
\caption{\label{tab:approximation_JMAP_JD} ERROR-RATES AT BOB WHEN USING JMAP DECODER AND JD DECODER FOR $M=2$}
\vspace{-0.25cm}
\begin{center}
\scalebox{0.85}{
\begin{tabular}{|ccccc|}
\hline
\multicolumn{5}{|c|}{$N_{C}=1$, $N_{B}=8$}                                                                                                                         \\ \hline
\multicolumn{1}{|c|}{SNR} & \multicolumn{1}{c|}{$\{\epsilon_{1},\epsilon_{2},\eta_{1},\eta_{2},\alpha\}$} & \multicolumn{1}{c|}{$P_{\text{JMAP}}$} & \multicolumn{1}{c|}{$P_{\text{JD}}$} & rel. error \\ \hline
\multicolumn{1}{|c|}{5 dB}    & \multicolumn{1}{c|}{$\{0, 1\text{e}^{-6},0.3052,2.6421, 0.4736\}$}& \multicolumn{1}{c|}{$3.06\times 10^{-1}$}& \multicolumn{1}{c|}{$3.23\times 10^{-1}$}& $5.55\times 10^{-2}$\\ 
\hline
\multicolumn{1}{|c|}{14 dB}    & \multicolumn{1}{c|}{$\{0,1\text{e}^{-6},0.5554,3.0750,0.8152\}$}& \multicolumn{1}{c|}{$8.32\times 10^{-2}$}& \multicolumn{1}{c|}{$8.42\times 10^{-2}$}& $1.20\times 10^{-2}$\\ 
\hline
\multicolumn{1}{|c|}{25 dB}    & \multicolumn{1}{c|}{$\{ 0,1\text{e}^{-6},0.4382,3.4008,0.9195\}$} & \multicolumn{1}{c|}{$1.88\times 10^{-2}$}& \multicolumn{1}{c|}{$1.90\times 10^{-2}$} & $1.06\times 10^{-2}$\\ 
\hline
\multicolumn{5}{|c|}{$N_{C}=2$, $N_{B}=4$}                                                                                                                        \\ \hline
\multicolumn{1}{|c|}{SNR} & \multicolumn{1}{c|}{$\{\epsilon_{1},\epsilon_{2},\eta_{1},\eta_{2},\alpha\}$} & \multicolumn{1}{c|}{$P_{\text{JMAP}}$} & \multicolumn{1}{c|}{$P_{\text{JD}}$} & rel. error \\ \hline
\multicolumn{1}{|c|}{5 dB}    & \multicolumn{1}{c|}{$\{ 0,1\text{e}^{-6},0.4334,2.7135,0.5734\}$}& \multicolumn{1}{c|}{$3.735\times 10^{-1}$}& \multicolumn{1}{c|}{$3.782\times 10^{-1}$}& $1.25\times 10^{-2}$\\ 
\hline
\multicolumn{1}{|c|}{14 dB}& \multicolumn{1}{c|}{$\{0,1\text{e}^{-6},0.5353,3.1645,0.8499\}$}& \multicolumn{1}{c|}{$1.32\times 10^{-1}$}            & \multicolumn{1}{c|}{$1.33\times 10^{-1}$}& $7.57\times 10^{-4}$ \\ \hline
\multicolumn{1}{|c|}{25 dB} & \multicolumn{1}{c|}{$\{0,1\text{e}^{-6},0.3228,3.6082,0.9655\}$}& \multicolumn{1}{c|}{$2.43\times 10^{-2}$}            & \multicolumn{1}{c|}{$2.47\times 10^{-2}$} & $1.64\times 10^{-2}$\\ \hline
\end{tabular}
}
\end{center}
\end{table}

\subsection{Joint Dominant (JD) Decoder for NC-FFFD Relaying Scheme}
\label{ssec:JD}
From \eqref{eq:def_l}, we observe that there exist a one-to-one correspondence between $(i, j)$ and $\ell$. Thus, the JD decoder in \eqref{eq:JD} can be rewritten as, $\hat{\ell} = \arg\underset{\ell \in\{1,\ldots, 2M\}}{\max\ }\tilde{g}\left(\mathbf{r}_{B}\vert \ell, e=0\right)$. Henceforth, a transmit pair jointly chosen by Alice and Charlie will be denoted by the index $\ell \in \{1, 2, \ldots, 2M\}$. As a consequence, the JD decoder only considers  the likelihood functions corresponding to the $2M$ dominant energy levels in $\mathcal{S}$ with the assumption that no decoding error is introduced by Charlie. Let $\bigtriangleup_{\substack{\ell\rightarrow \ell^{*}\\ \ell \neq \ell^{*}}}$ denotes the event when Bob incorrectly decodes an index $\ell$ to $\ell^{*}$ such that $\ell \neq \ell^{*}$. Then, $\Pr\left(\bigtriangleup_{\substack{\ell\rightarrow \ell^{*}\\ \ell \neq \ell^{*}}}\right)=\Pr\left(\tilde{g}\left(\mathbf{r}_{B}\vert\ell, e=0\right)\leq \tilde{g}\left(\mathbf{r}_{B}\vert \ell^{*}, e=0\right)\right)$.
To characterize $\Pr\left(\bigtriangleup_{\substack{\ell\rightarrow \ell^{*}\\ \ell \neq \ell^{*}}}\right)$, one should determine the energy detection threshold between the energy levels corresponding to ${\ell}$ and ${\ell^{*}}$. Towards this direction, we use the following lemma that computes the energy detection threshold between $S_{\ell}$ and $S_{\ell^{*}}$.

\begin{lemma}
\label{lm:rho}
If $S_{\ell}$ denotes the energy level jointly contributed by Alice and Charlie corresponding to the transmitted index $\ell$ and $S_{\ell^{*}}$ denotes the energy level corresponding to the decoded index $\ell^{*}$ such that $\ell \neq \ell^{*}$, then the probability of the event $\bigtriangleup_{\substack{\ell\rightarrow \ell^{*}\\ \ell \neq \ell^{*}}}$ is given by $\Pr\left(\bigtriangleup_{\substack{\ell\rightarrow \ell^{*}\\ \ell \neq \ell^{*}}}\right) = \Pr(\mathbf{r}_{B}^{H}\mathbf{r}_{B} \geq \rho_{\ell,\ell^{*}})$, where the threshold $\rho_{\ell,\ell^{*}}$ is given by, $\rho_{\ell,\ell^{*}} \approx N_{B}\frac{S_{\ell}S_{\ell^{*}}}{S_{\ell^{*}}-S_{\ell}}\ln\left(\frac{S_{\ell^{*}}}{S_{\ell}}\right)$. 
\end{lemma}

\begin{proof}
Bob uses JD decoder and compares the conditional PDF of $\mathbf{r}_{B}$ conditioned on $\ell$ and $\ell^{*}$ as,
\bieee
\Pr\left(\bigtriangleup_{\substack{\ell\rightarrow \ell^{*}\\ \ell \neq \ell^{*}}}\right) &=& \Pr\left(\dfrac{\tilde{g}\left(\mathbf{r}_{B}\vert \ell^{*}, e=0\right)}{\tilde{g}\left(\mathbf{r}_{B}\vert \ell, e=0\right)}\leq 1\right)=
\Pr\left(\dfrac{\dfrac{P_{i^{*}i^{*}}}{\left(\pi S_{\ell^{*}}\right)^{N_{B}}}\exp\left(-\frac{\mathbf{r}_{B}^{H}\mathbf{r}_{B}}{S_{\ell^{*}}}\right)}{\dfrac{P_{ii}}{\left(\pi S_{\ell}\right)^{N_{B}}}\exp\left(-\frac{\mathbf{r}_{B}^{H}\mathbf{r}_{B}}{S_{\ell}}\right)}\leq 1\right)= \Pr\left(\mathbf{r}_{B}^{H}\mathbf{r}_{B}\geq \rho_{\ell,\ell^{*}}\right),\nn
\eieee
\noindent where $\rho_{\ell,\ell^{*}}=\frac{S_{\ell}S_{\ell^{*}}}{S_{\ell^{*}}-S_{\ell}}\left(N_{B}\ln\left(\frac{S_{\ell^{*}}}{S_{\ell}}\right) + \ln\left(\frac{P_{ii}}{P_{i^{*}i^{*}}}\right)\right)$ and $P_{ii}$ and $P_{i^{*}i^{*}}$ are a priori probabilities of index $\ell$ and $\ell^{*}$, respectively. It is straightforward that when $i=i^{*}$, $\ln\left(\frac{P_{ii}}{P_{i^{*}i^{*}}}\right)=0$. Further, since $\ln\left(\frac{P_{00}}{P_{11}}\right)\approx 0$ for $N\geq 1$, when $i\neq i^{*}$, we have $\ln\left(\frac{P_{ii}}{P_{i^{*}i^{*}}}\right)=\ln\left(\frac{P_{00}}{P_{11}}\right)\approx 0$ and $\ln\left(\frac{P_{ii}}{P_{i^{*}i^{*}}}\right)=\ln\left(\frac{P_{11}}{P_{00}}\right)\approx 0$, for $i=0$ and $i=1$, respectively. Thus, $\rho_{\ell,\ell^{*}}\approx N_{B}\frac{S_{\ell}S_{\ell^{*}}}{S_{\ell^{*}}-S_{\ell}}\ln\left(\frac{S_{\ell^{*}}}{S_{\ell}}\right)$.
\end{proof}

Since $S_{1}<S_{2}\cdots<S_{2M-1}<S_{2M}$, the set of relevant thresholds for the JD decoder are $\{\rho_{\ell,\ell + 1}, \ell = 1, 2, \ldots, 2M-1\}$. Therefore, based on the  received energy $\mathbf{r}_{B}^{H}\mathbf{r}_{B}$, the JD decoder for detecting $\hat{\ell}$ can be realized using an energy detector as, $\hat{\ell}=\ell$, if  $\rho_{\ell - 1,\ell} < \mathbf{r}_{B}^{H}\mathbf{r}_{B} \leq \rho_{\ell,\ell + 1 }$, where $\rho_{0,1}=0$ and $\rho_{2M,\infty}=\infty$.

Using $\hat{\ell}$, the average Symbol Error Probability (SEP), denoted by $P_{e}$, is given by, $P_{e} = \frac{1}{2M} \sum_{\ell = 1}^{2M} P_{e, \ell}$, where $P_{e,\ell} = \Pr\left(\ell\neq\ell^{*}\right)$ is the probability that Bob decodes a transmitted index $\ell$ as $\ell^{*}$, where $\ell\neq\ell^{*}$. Since, the decision of the energy detector is based on the received energies at Bob, we notice that sum energy levels can be from $\mathcal{S}$, when $e=0$ or $\mathcal{\overline{S}}$, when $e=1$. Therefore, $P_{e,\ell} = \Pr(e=0)\Pr\left(\ell\neq\ell^{*}\vert e=0\right) + \Pr(e=1)\Pr\left(\ell\neq\ell^{*}\vert e=1\right)$. Thus, we have
\begin{equation}
P_{e, \ell}=
\begin{cases}
P_{00}P_{e, S_{\ell}} + P_{01}P_{e, \overline{S}_{\ell}} & \text{if }\ell(\mathrm{mod}4)\leq 1,
\\
P_{11}P_{e, S_{\ell}} + P_{10}P_{e, \overline{S}_{\ell}} & \text{if } \text{otherwise},
\end{cases}
\label{eq:Pe_formal2}
\end{equation}
\noindent where $P_{e, S_{\ell}}$ and $P_{e, \overline{S}_{\ell}}$ are the terms associated with erroneous decision in decoding $\ell$, when $e=0$ and $e=1$, respectively. Since $\mathbf{r}_{B}^{H}\mathbf{r}_{B}$ is gamma distributed, we get $P_{e, S_{\ell}}$ as given in \eqref{eq:errors_dominant}. 
\begin{small}
\begin{equation}
P_{e,S_{\ell}}=
\begin{cases}
1-\Pr\left(\mathbf{r}_{B}^{H}\mathbf{r}_{B}\leq\rho_{1, 2}\vert e=0\right) = \dfrac{\Gamma\left(N_{B}, \frac{\rho_{1,2}}{S_{1}}\right)}{\Gamma\left(N_{B}\right)} & \text{for }\ell=1,
\\
1-\Pr\left(\rho_{\ell-1,\ell}\leq\mathbf{r}_{B}^{H}\mathbf{r}_{B}\leq\rho_{\ell, \ell+1}\vert e=0\right) = \dfrac{\gamma\left(N_{B}, \frac{\rho_{\ell-1,\ell}}{S_{\ell}}\right)}{\Gamma\left(N_{B}\right)} + \dfrac{\Gamma\left(N_{B}, \frac{\rho_{\ell,\ell+1}}{S_{\ell}}\right)}{\Gamma\left(N_{B}\right)} & \text{for } 2\leq\ell\leq 2M-1,
\\
1-\Pr\left(\mathbf{r}_{B}^{H}\mathbf{r}_{B}>\rho_{2M-1,2M}\vert e=0\right) = \dfrac{\gamma\left(N_{B}, \frac{\rho_{2M-1,2M}}{S_{2M}}\right)}{\Gamma\left(N_{B}\right)} & \text{for } \ell=2M.
\end{cases}
\label{eq:errors_dominant}
\end{equation}
\end{small}
\noindent Since Bob uses the same thresholds to compute $P_{e, \overline{S}_{\ell}}$, we obtain the expression of $P_{e, \overline{S}_{\ell}}$, by replacing $S_{\ell}$ by $\overline{S}_{\ell}$ in \eqref{eq:errors_dominant}. Finally, substituting \eqref{eq:Pe_formal2}, \eqref{eq:errors_dominant}, and corresponding $P_{e, \overline{S}_{\ell}}$ in $P_{e}$, we get,
\begin{multline}
P_{e} = \frac{1}{2M}\left[ \sum_{\ell_{1} = 1}^{M}\left(P_{00}P_{e, S_{\frac{1}{2}\left(4\ell_{1}+(-1)^{\ell_{1}}-1\right)}} + P_{01}P_{e, \overline{S}_{\frac{1}{2}\left(4\ell_{1}+(-1)^{\ell_{1}}-1\right)}}\right)\right.\\
\ \left. + \sum_{\ell_{2} = 1}^{M}\left(P_{11}P_{e, S_{\frac{1}{2}\left((-1)^{\ell_{2}}\left(4(-1)^{\ell_{2}}\ell_{2} + (-1)^{\ell_{2}+1}-1\right)\right)}} + P_{10}P_{e, \overline{S}_{\frac{1}{2}\left((-1)^{\ell_{2}}\left(4(-1)^{\ell_{2}}\ell_{2} + (-1)^{\ell_{2}+1}-1\right)\right)}}\right)\right].\label{eq:Pe} 
\end{multline}

\section{Optimization of Energy Levels}
\label{sec:optimization}
In this section, we formulate an optimization problem in order to compute the optimal energy levels at Alice and Charlie. In particular, as given in \eqref{opt}, we fix $N_{C}$ and $N_{B}$ and then optimise the energy levels, $\{\epsilon_{j},\eta_{j}\}$, and $\alpha$ so as to minimise the SEP subject to the  energy constraint in \eqref{eq:new_constaint}.
\begin{mdframed}
\bieee
\underset{\epsilon_{1},\cdots,\epsilon_{M}, \eta_{1},\cdots,\eta_{M}, \alpha}{\min} \quad &  & P_{e}\label{opt}\\
 \text{subject to:} \quad & &\sum_{j=1}^{M}(\epsilon_{j}+\eta_{j}) = M(1+\alpha), \epsilon_{1}<\cdots<\epsilon_{M}, \eta_{1}<\cdots<\eta_{M}, 0<\alpha<1, \nn\\
 & & \epsilon_{j}<\eta_{j} \text{ for }j\in\{1,3,\cdots, 2M-1\}, \epsilon_{j}>\eta_{j} \text{ for } j\in\{2,4,\cdots, 2M\}.\nn
\eieee 
\end{mdframed}
\noindent One can solve the above optimization problem by first formulating the Lagrangian and then solving the system of $2M+2$ non-linear equations. Since solving a system of non-linear equations is complex in general, we use an alternate approach for minimising $P_{e}$ using its analytical structure, as discussed in the next section. We first discuss the optimization of energy levels for $M=2$ and then propose a generalised approach of $M=2^{m}$ such that $m > 1$.

\subsection{Optimization of Energy Levels for $M=2$}
\label{ssec:Globecom}
The expression of SEP in \eqref{eq:Pe} when $M=2$ is given as,
\bieee
 P_{e}\! =\! \dfrac{1}{4}\left(P_{00}\left(P_{e,S_{1}}\! +\!P_{e,S_{4}}\right) \!+\! P_{11}\left(P_{e,S_{2}}\! +\!P_{e,S_{3}}\right)\! +\! P_{01}\left(P_{e,\overline{S}_{1}}\! +\!P_{e,\overline{S}_{4}}\right)\! +\! P_{10}\left(P_{e,\overline{S}_{2}}\! +\! P_{e,\overline{S}_{3}}\right)\right).\label{eq:Pe_M2}
\eieee
Instead of using $P_{e}$ for optimization problem, we use an upper-bound on $P_{e}$, where we upper-bound $P_{e,\overline{S}_{1}}\!\leq\! P_{e,\overline{S}_{4}}\!\leq\! P_{e,\overline{S}_{2}}\!\leq \! P_{e,\overline{S}_{3}}\!\leq\! 1$, such that,
\bieee
 P_{e}\leq P_{e}^{\prime}\triangleq \dfrac{1}{4}\left(P_{00}\left(P_{e,S_{1}}\! +\!P_{e,S_{4}}\right) \!+\! P_{11}\left(P_{e,S_{2}}\! +\!P_{e,S_{3}}\right)\! +\! 2\left(P_{01}+P_{10}\right)\right).\label{eq:Pe_M2U}
\eieee
\noindent Henceforth, we optimise the energy levels, $\epsilon_{1}$, $\epsilon_{2}$, $\eta_{1}$, and $\eta_{2}$ and $\alpha$ so as to minimise $P_{e}^{\prime}$.\footnote{Later through simulation results, we show that, optimizing \eqref{eq:Pe_M2U} gives us near-optimal results.} Thus, the modified optimization problem when $M=2$ is,
\bieee
\underset{\epsilon_{1},\epsilon_{2}, \eta_{1},\eta_{2}, \alpha}{\min} \quad &  & P_{e}^{\prime}\label{opt:M2}\\
 \text{subject to:} \quad & &\epsilon_{1}+\epsilon_{2}+\eta_{1}+\eta_{2} = 2(1+\alpha), \epsilon_{1}<\epsilon_{2}, \eta_{1}<\eta_{2},0<\alpha<1, \epsilon_{1}<\eta_{1}<\eta_{2}<\epsilon_{2}.\nn
\eieee 
In order to minimise $P_{e}^{\prime}$, it is clear that we must minimise each $P_{e,S_{\ell}}$, for $\ell=1,\cdots,4$ in \eqref{opt:M2}. Towards this direction, in the next lemma, we show that when $\epsilon_{1}=0$, $P_{e,S_{1}}$ is minimum.

\begin{lemma}\label{lm:epsilon1}
The expression $P_{e,S_{1}} = \dfrac{\Gamma\left(N_{B}, \frac{\rho_{1,2}}{S_{1}}\right)}{\Gamma\left(N_{B}\right)}$ is minimum when $\epsilon_{1}=0$.
\end{lemma}
\begin{proof}
The expression of $P_{e,S_{1}}$ is an upper incomplete Gamma function.  Since upper incomplete Gamma function is a decreasing function of the second parameter, $\Gamma\left(N_{B}, \frac{\rho_{1,2}}{S_{1}}\right)$ is a decreasing function of $\frac{\rho_{1,2}}{S_{1}}$. Therefore, $P_{e,S_{1}}$ is minimum when $\frac{\rho_{1,2}}{S_{1}}$ is maximum and $\frac{\rho_{1,2}}{S_{1}}$ is maximum when $S_{1}$ is minimum. Since $S_{1}=\epsilon_{1}+N_{o}$, $S_{1}$ is minimum when $\epsilon_{1}=0$. This completes the proof.
\end{proof}
\begin{lemma}
\label{lm:P12P21}
At high SNR, $P_{e,S_{1}}\ll 1$ and $P_{e,S_{2}}\approx  \dfrac{\Gamma\left(N_{B}, \frac{\rho_{2,3}}{S_{2}}\right)}{\Gamma\left(N_{B}\right)}$. 
\end{lemma}
\begin{proof}
We first prove that $P_{e,S_{1}}\ll 1$. We have $P_{e,S_{1}}=\frac{\Gamma\left(N_{B}, \frac{\rho_{1,2}}{S_{1}}\right)}{\Gamma\left(N_{B}\right)}$. The ratio $\frac{\rho_{1,2}}{S_{1}}$ is expressed as, $N_{B}\frac{\ln(1+\kappa_{1})}{\kappa_{1}}$, where $\kappa_{1}=(S_{1}-S_{2})/S_{2}$. Further. since $S_{1}<S_{2}$, $-1<\kappa_{1}<0$. Also, the ratio $\frac{\ln(1+\kappa_{1})}{\kappa_{1}}$ follows the inequalities, $\frac{2}{2+\kappa_{1}}\leq\frac{\ln(1+\kappa_{1})}{\kappa_{1}}\leq \frac{2+\kappa_{1}}{2+2\kappa_{1}}$, for $\kappa > -1$. Therefore, $\frac{\Gamma\left(N_{B}, \frac{2N_{B}}{2+\kappa_{1}}\right)}{\Gamma\left(N_{B}\right)}\geq\frac{\Gamma\left(N_{B}, \frac{\rho_{1,2}}{S_{1}}\right)}{\Gamma\left(N_{B}\right)}\geq \frac{\Gamma\left(N_{B}, N_{B}\frac{2+\kappa_{1}}{2+2\kappa_{1}}\right)}{\Gamma\left(N_{B}\right)}$, where the second inequality is because $\Gamma\left(N_{B}, \frac{\rho_{1,2}}{S_{1}}\right)$ is a decreasing function of $\frac{\rho_{1,2}}{S_{1}}$. Thus, $\frac{\Gamma\left(N_{B}, \frac{\rho_{1,2}}{S_{1}}\right)}{\Gamma\left(N_{B}\right)}\leq \frac{\Gamma\left(N_{B}, \frac{2N_{B}}{2+\kappa_{1}}\right)}{\Gamma\left(N_{B}\right)} = \frac{\Gamma\left(N_{B}, 2N_{B}\right)}{\Gamma\left(N_{B}\right)}\ll 1$. Since $S_{1}\approx 0$ at high SNR, $2/(2+\kappa_{1}) = 2S_{2}/(S_{1}+S_{2})\approx 2$ and therefore, we have the second inequality. This proves the first part of Lemma. On similar lines, we can prove that at high SNR, the term $\frac{\gamma\left(N_{B}, \frac{\rho_{1,2}}{S_{2}}\right)}{\Gamma\left(N_{B}\right)}\leq\frac{\gamma\left(N_{B}, \frac{N_{B}}{2}\right)}{\Gamma\left(N_{B}\right)}$, thus, $\frac{\gamma\left(N_{B}, \frac{N_{B}}{2}\right)}{\Gamma\left(N_{B}\right)}\ll 1$ and therefore, we have $P_{e,S_{2}} \approx \frac{\Gamma\left(N_{B}, \frac{\rho_{2,3}}{S_{2}}\right)}{\Gamma\left(N_{B}\right)}$.
\end{proof}
Using the results of Lemma~\ref{lm:P12P21}, the expression of $P_{e}^{\prime}$ is approximated as,
\bieee
P_{e}^{\prime}\approx\dfrac{1}{4}\left(P_{00}P_{e,S_{4}} \!+\! P_{11}\left(P_{e,S_{2}}\! +\!P_{e,S_{3}}\right)\! +\! 2\left(P_{01}+P_{10}\right)\right).\label{eq:Pe_app}
\eieee
From \eqref{opt:M2} we have 5 variables, resulting in a 5-dimensional search space to find the optimal set $\{\epsilon_{1},\epsilon_{2},\eta_{1},\eta_{2},\alpha\}$. Using the result of Lemma~\ref{lm:epsilon1}, we have $\epsilon_{1}=0$. Further, rearranging the sum energy constraint, we express $\epsilon_{2}$ as a function of $\eta_{1}$, $\eta_{2}$, and $\alpha$, therefore, $\epsilon_{2} = 2(1+\alpha)-(\eta_{1}+\eta_{2})$. Thus, the search space is reduced to 3 dimensions. Through simulations we observe that, when we fix $\eta_{1}$ and $\alpha$, $P_{e}^{\prime}$ exhibits unimodal nature w.r.t. $\eta_{2}$. Similarly, $P_{e}^{\prime}$ is unimodal w.r.t. $\alpha$, when we fix $\eta_{1}$ and $\eta_{2}$. The variation of $P_{e}^{\prime}$, the increasing terms of $P_{e}^{\prime}$, and the decreasing terms of $P_{e}^{\prime}$, w.r.t. $\eta_{2}$ and $\alpha$ are shown in Fig.~\ref{fig:unimodal_eta2} and Fig.~\ref{fig:unimodal_alpha}, respectively. Further, we also observe that the unique mode in both the cases is very close to the intersection of increasing and decreasing terms of $P_{e}^{\prime}$. Therefore, in the next two theorems, we prove that the increasing and decreasing terms of $P_{e}^{\prime}$ w.r.t. $\eta_{2}$ and $\alpha$, have a unique intersection that is close to the local minima of $P_{e}^{\prime}$.

\begin{figure}[!htb]
\vspace{-0.25in}
    \centering
    \begin{minipage}[t]{.48\textwidth}
        \centering
        \includegraphics[width = 0.66\textwidth, height = 0.6\linewidth]{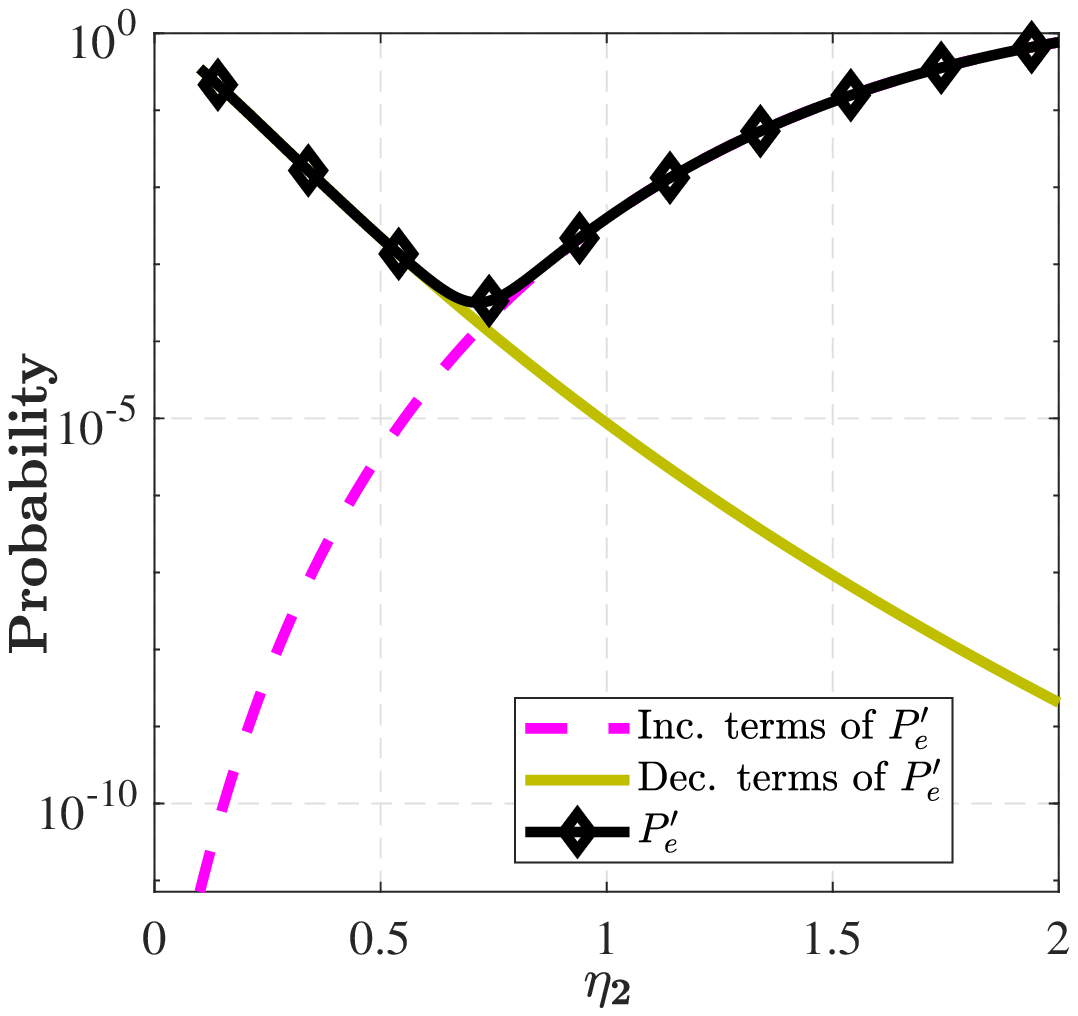}
        \caption{\label{fig:unimodal_eta2}  Variation of $P_{e}^{\prime}$, its increasing and decreasing terms as a function of $\eta_{2}$, when $\eta_{1}$ and $\alpha$ are fixed.}
    \end{minipage}%
        \hfill
    \begin{minipage}[t]{0.48\textwidth}
        \centering
        \includegraphics[width = 0.66\textwidth, height = 0.6\linewidth]{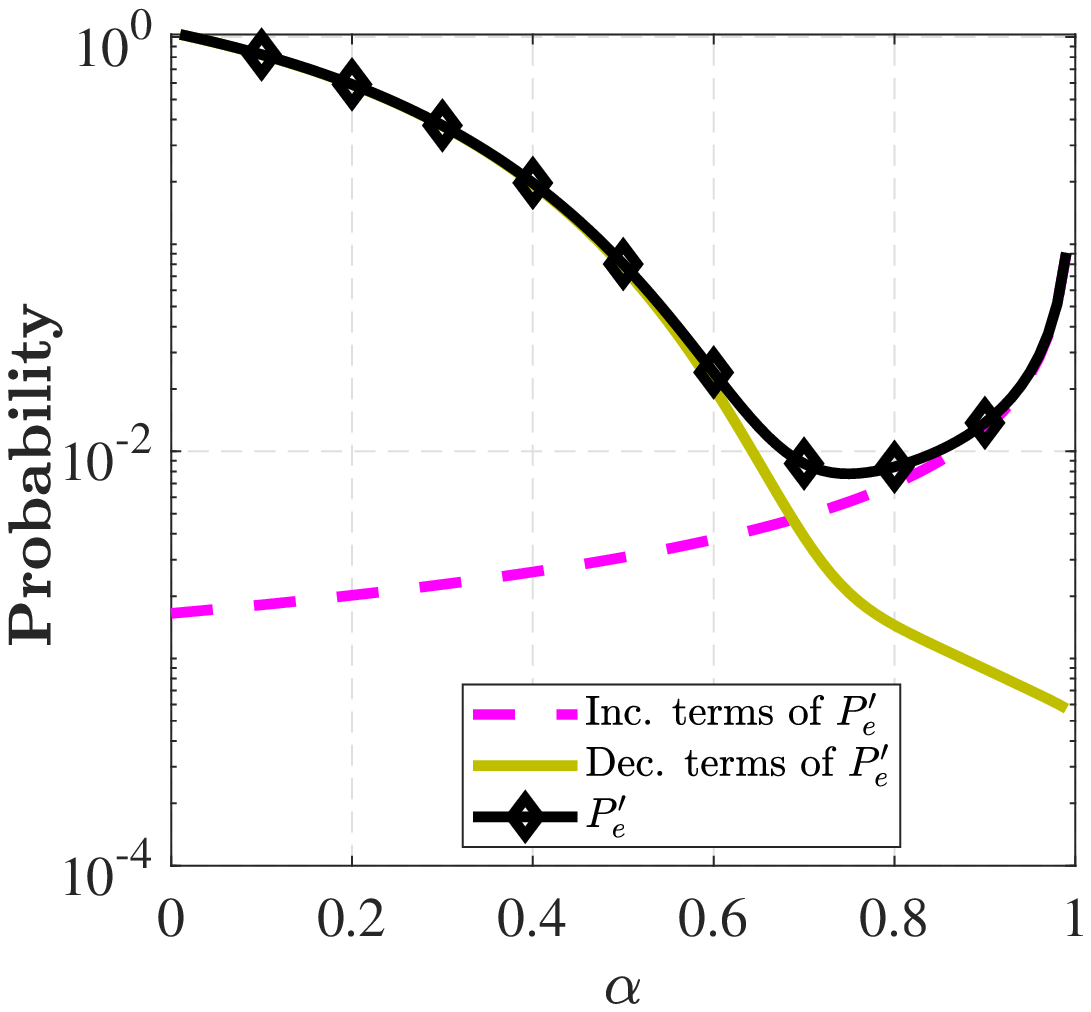}
\caption{\label{fig:unimodal_alpha}Variation of $P_{e}^{\prime}$, its increasing and decreasing terms as a function of $\alpha$, when $\eta_{1}$ and $\eta_{2}$ are fixed.}
    \end{minipage}
\end{figure}
\begin{theorem}
\label{th:Pe_eta2}
For a given $\eta_{1}$ and $\alpha$, the increasing and decreasing terms in $P_{e}^{\prime}$ intersect only once for $\eta_{2}\in\left(\eta_{1},1+\alpha-0.5\eta_{1}\right)$. 
\end{theorem}
\begin{proof}
We first determine the increasing and decreasing terms of $P_{e}^{\prime}$. Towards this direction, we first analyse the behaviour of each term in \eqref{eq:Pe_app}, i.e., $P_{e,S_{2}}$, $P_{c,S_{3}}$, and $P_{e,S_{4}}$ as a function of $\eta_{2}$, where
\bieee
P_{e,B2} = \dfrac{\Gamma\left(N_{B}, \frac{\rho_{2,3}}{S_{2}}\right)}{\Gamma\left(N_{B}\right)},\ \ P_{e,S_{3}} = \dfrac{\Gamma\left(N_{B}, \frac{\rho_{3,4}}{S_{3}}\right)}{\Gamma\left(N_{B}\right)} + \dfrac{\gamma\left(N_{B}, \frac{\rho_{2,3}}{S_{3}}\right)}{\Gamma\left(N_{B}\right)},\ \ P_{c,S_{4}} = \dfrac{\gamma\left(N_{B}, \frac{\rho_{3,4}}{S_{4}}\right)}{\Gamma\left(N_{B}\right)}.\nn
\eieee
\noindent Consider the term $P_{e,S_{2}}$, where the ratio $\frac{\rho_{2,3}}{S_{2}}$ is given by, $N_{B}\frac{\ln(1+\kappa_{3})}{\kappa_{3}}$, where $\kappa_{3}=(S_{2}-S_{3})/S_{3}$. Since $S_{2}<S_{3}$,  $\kappa_{3}<0$. Differentiating $\kappa_{3}$ w.r.t. $\eta_{2}$ we get $-S_{1}/S_{2}^{2}$. Therefore, as $\eta_{2}$ increases, $\kappa_{3}$ decreases. Since $\ln(1+\kappa_{3})/\kappa_{3}$ is a decreasing function of $\kappa_{3}$, as $\kappa_{3}$ decreases, $N_{B}{\ln(1+\kappa_{3})}/{\kappa_{3}}$ increases. Finally, since $\frac{\Gamma\left(N_{B}, \frac{\rho_{2,3}}{S_{2}}\right)}/{\Gamma\left(N_{B}\right)}$ is a decreasing function of ${\rho_{2,3}}/{S_{2}}$, $P_{e,B2}$ decreases with increasing ${\ln(1+\kappa_{3})}/{\kappa_{3}}$. Therefore, $P_{e,S_{4}}$ is a decreasing function of $\eta_{2}$.

On similar lines, we can prove that $\frac{\gamma\left(N_{B}, \frac{\rho_{2,3}}{S_{3}}\right)}{\Gamma\left(N_{B}\right)}$ is also a decreasing function of $\eta_{2}$. In contrast, the terms, $\frac{\Gamma\left(N_{B}, \frac{\rho_{3,4}}{S_{3}}\right)}{\Gamma\left(N_{B}\right)}$ and $\frac{\gamma\left(N_{B}, \frac{\rho_{3,4}}{S_{4}}\right)}{\Gamma\left(N_{B}\right)}$ are increasing functions of $\eta_{2}$.

To prove that the increasing and decreasing terms intersect only once, we can prove that the order of increasing and decreasing terms reverses at extreme values of $\eta_{2}\in(\eta_{1}, (1+\alpha-0.5\eta_{1}))$. Thus, we evaluate the sum of decreasing terms at left extreme, i.e., $\eta_{2}\rightarrow\eta_{1}$ and right extreme, i.e., $\eta_{2}\rightarrow(1+\alpha-0.5\eta_{1})$,
\bieee
\lim_{\eta_{2}\rightarrow\eta_{1}}\dfrac{\Gamma\left(N_{B}, \frac{\rho_{2,3}}{S_{2}}\right)}{\Gamma\left(N_{B}\right)} + \frac{\Gamma\left(N_{B}, \frac{\rho_{2,3}}{S_{3}}\right)}{\Gamma\left(N_{B}\right)} = 1 \text{ and } \lim_{\eta_{2}\rightarrow(1+\alpha-0.5\eta_{1})}\frac{\Gamma\left(N_{B}, \frac{\rho_{2,3}}{S_{2}}\right)}{\Gamma\left(N_{B}\right)} + \frac{\Gamma\left(N_{B}, \frac{\rho_{2,3}}{S_{3}}\right)}{\Gamma\left(N_{B}\right)} \ll 1.\nn
\eieee
\noindent Similarly, we evaluate the sum of increasing terms at left extreme and right extremes of $\eta_{1}$,
\bieee
\lim_{\eta_{2}\rightarrow\eta_{1}}\frac{\Gamma\left(N_{B}, \frac{\rho_{3,4}}{S_{3}}\right)}{\Gamma\left(N_{B}\right)} + \frac{\gamma\left(N_{B}, \frac{\rho_{3,4}}{S_{4}}\right)}{\Gamma\left(N_{B}\right)} \ll 1, \text{ and }\ \lim_{\eta_{2}\rightarrow(1+\alpha-0.5\eta_{1})} \frac{\Gamma\left(N_{B}, \frac{\rho_{3,4}}{S_{3}}\right)}{\Gamma\left(N_{B}\right)} + \frac{\gamma\left(N_{B}, \frac{\rho_{3,4}}{S_{4}}\right)}{\Gamma\left(N_{B}\right)} = 1.\nn
\eieee

The above discussion is summarised as,
\begin{equation*}
\begin{cases}
\dfrac{\Gamma\left(N_{B}, \frac{\rho_{2,3}}{S_{2}}\right)}{\Gamma\left(N_{B}\right)} + \dfrac{\Gamma\left(N_{B}, \frac{\rho_{2,3}}{S_{3}}\right)}{\Gamma\left(N_{B}\right)} > \dfrac{\Gamma\left(N_{B}, \frac{\rho_{3,4}}{S_{3}}\right)}{\Gamma\left(N_{B}\right)} + \dfrac{\gamma\left(N_{B}, \frac{\rho_{3,4}}{S_{4}}\right)}{\Gamma\left(N_{B}\right)}, & \text{if $\eta_{2}\rightarrow\eta_{1}$},\\
\dfrac{\Gamma\left(N_{B}, \frac{\rho_{2,3}}{S_{2}}\right)}{\Gamma\left(N_{B}\right)} + \dfrac{\Gamma\left(N_{B}, \frac{\rho_{2,3}}{S_{3}}\right)}{\Gamma\left(N_{B}\right)} < \dfrac{\Gamma\left(N_{B}, \frac{\rho_{3,4}}{S_{3}}\right)}{\Gamma\left(N_{B}\right)} + \dfrac{\gamma\left(N_{B}, \frac{\rho_{3,4}}{S_{4}}\right)}{\Gamma\left(N_{B}\right)}, & \text{if $\eta_{2}\rightarrow(1+\alpha-0.5\eta_{1})$}.
\end{cases}
\end{equation*}
\end{proof}

\begin{theorem}
\label{th:Pe_alpha}
For a given $\eta_{1}$ and $\eta_{2}$, the increasing and decreasing terms in $P_{e}^{\prime}$ intersect only once for $\alpha\in\left(0,1\right)$. 
\end{theorem}
\begin{proof}
Since $\alpha$ is variable, we recall Lemma~\ref{lm:P10P01_alpha} to show that $P_{01}$ and $P_{10}$ are decreasing function of $\alpha$. Further, since $P_{01}$ and $P_{10}$ are decreasing functions of $\alpha$, $P_{00}$ and $P_{11}$ are decreasing functions of $\alpha$. In addition to these $4$ probabilities, $P_{e,S_{2}}$, $P_{e,S_{3}}$, and $P_{e,S_{4}}$ are also functions of $\alpha$ in \eqref{eq:Pe_app}. On similar lines of Theorem~\ref{th:Pe_eta2}, we prove that, $P_{e,S_{2}}$, $P_{e,S_{3}}$, and $P_{e,S_{4}}$ are decreasing function of $\alpha$. Therefore, we observe that $P_{00}P_{e,S_{4}}+ P_{11}\left(P_{e,S_{2}} + P_{e,S_{3}}\right)$ is a decreasing function of $\alpha$ and since $P_{00}=P_{11}\approx 0$, when $\alpha\rightarrow 1$, $P_{00}P_{e,S_{4}}+ P_{11}\left(P_{e,S_{2}} + P_{e,S_{3}}\right)\approx 0$, when $\alpha\rightarrow 1$. Further, $2(P_{01}+P_{10})$ is an increasing function of $\alpha$ such that, $2(P_{01}+P_{10})\approx 0$, when $\alpha\rightarrow 0$ and   $2(P_{01}+P_{10})\approx 2$, when $\alpha\rightarrow 1$. Therefore, it is straightforward to observe that the increasing and decreasing terms of $P_{e}^{\prime}$ reverse their orders at extreme values of $\alpha$. Thus, they have a unique intersection point.
\end{proof}

In the next section, we use Theorem~\ref{th:Pe_eta2} and Theorem~\ref{th:Pe_alpha} to present a low-complexity algorithm to solve the optimization problem in \eqref{opt:M2}. Using this algorithm, we obtain a local minima over the variables $\eta_{2}$ and $\alpha$ for a given $\eta_{1}$.

\subsubsection{Two-Layer Greedy Descent (TLGD) Algorithm}
In this section, we present Two-Layer Greedy Descent (TLGD) algorithm, as presented in Algorithm~\ref{Algo:M2}. It first fixes $N_{C}$, $N_{B}$, and SNR and then initialise $\eta_{1} = 0$, and $\eta_{2}$ and $\alpha$ with arbitrary values $\eta_{2}^{o}$ and $\alpha^{o}$, respectively. Using the initial values, it computes $P_{e}^{o}$ using \eqref{eq:Pe_app} and then obtains $\eta_{2}^{i}$ and $\alpha^{i}$ using Theorem~\ref{th:Pe_eta2} and Theorem~\ref{th:Pe_alpha}, respectively. It then evaluates $P_{e}^{\eta_{2}}$, i.e., $P_{e}^{\prime}$ at $\left\{\eta_{1}, \eta_{2}^{i}, \alpha\right\}$ and $P_{e}^{\alpha}$, i.e., $P_{e}^{\prime}$ at $\left\{\eta_{1}, \eta_{2}, \alpha^{i}\right\}$. If for a given $\eta_{1}$,  $\left\vert P_{e}^{\alpha}-P_{e}^{\eta_{2}}\right\vert < \delta_{P_{e}^{\prime}}$, for some $\delta_{P_{e}^{\prime}}>0$, then the algorithm exits the inner while-loop with $P_{e}^{\iota}$ such that $P_{e}^{\iota} = \min\left(P_{e}^{\alpha}, P_{e}^{\eta_{2}}\right)$ else, the algorithm iteratively descents in the steepest direction with new values of $\eta_{2}$ and $\alpha$. After traversing several values of $\eta_{1}$, TLGD finally stops when for a given $\eta_{1}$, the obtained $P_{e}^{\iota}$ is within $\delta_{P_{e}^{\prime}}$, resolution of the previously computed value. The points at which $P_{e}^{\prime}$ is minimum as computed by TLGD are given by $\eta_{1}^{\star}$, $\eta_{2}^{\star}$ and $\alpha^{\star}$. We rearrange the constraint in~\eqref{opt:M2} to obtain $\epsilon_{2}^{\star}=2(1+\alpha^{\star})-\left(\eta_{1}^{\star} + \eta_{2}^{\star}\right)$. Further, from Lemma~\ref{lm:epsilon1}, we have $\epsilon_{1}=0$, therefore, $\epsilon_{1}^{\star}=0$. Thus, TLGD computes all the 5 variables, i.e., $\epsilon_{1}^{\star}$, $\epsilon_{2}^{\star}$, $\eta_{1}^{\star}$, $\eta_{2}^{\star}$, and $\alpha^{\star}$.
\begin{algorithm}
\setstretch{0.33}
\DontPrintSemicolon 
  \KwInput{$P_{e}^{\prime}$ from~\eqref{eq:Pe_app}, $\delta_{P_{e}^{\prime}}>0$, $\delta_{\eta_{1}}>0$, $\epsilon_{1}=0$}
  \KwOutput{$\left\{\eta_{1}^{\star}, \eta_{2}^{\star},\alpha^{\star}\right\}$}
  Initialize: $\eta_{1}\gets 0$, $\eta_{2}\gets \eta_{2}^{o}$, $\alpha\gets \alpha^{o}$\\
     $P_{e}^{o} \gets P_{e}^{\prime}\left(\alpha,\eta_{1},\eta_{2}\right)$\\
  \While{true} 
  {
  \While{true}
  {
  Compute $\eta_{2}^{i}$ using Theorem~\ref{th:Pe_eta2} and obtain $P_{e}^{\eta_{2}} \gets P_{e}^{\prime}\left(\eta_{1}, \eta_{2}^{i},\alpha\right)$\\
 Compute $\alpha^{i}$ using Theorem~\ref{th:Pe_alpha} and obtain $P_{e}^{\alpha} \gets P_{e}^{\prime}\left(\eta_{1}, \eta_{2},\alpha^{i}\right)$\\
  \If{$P_{e}^{\alpha}-P_{e}^{\eta_{2}} \geq \delta_{P_{e}^{\prime}}$}
  {
  $\eta_{2} \gets \eta_{2}^{i}$; continue
  }
  \ElseIf{$P_{e}^{\alpha}-P_{e}^{\eta_{2}} \leq -\delta_{P_{e}^{\prime}}$}
  {
  $\alpha \gets \alpha^{i}$; continue
  }
  \ElseIf {$\left\vert P_{e}^{\alpha}-P_{e}^{\eta_{2}}\right\vert<\delta_{P_{e}^{\prime}}$}
  {
  $P_{e}^{\iota} = \min\left(P_{e}^{\alpha}, P_{e}^{\eta_{2}}\right)$; break
  }
  }
   \If{$\left(P_{e}^{\iota}-P_{e}^{o}\right) \leq- \delta_{P_{e}^{\prime}}$}
  {
  $\eta_{1} \gets \eta_{1} + \delta_{\eta_{1}}$, $P_{e}^{o}\gets P_{e}^{\iota}$; $\alpha^{\ast}\gets \alpha$, $\eta_{2}^{\ast}\gets \eta_{2}$
  }
  \ElseIf{$\left(P_{e}^{\iota}-P_{e}^{o}\right) \geq \delta_{P_{e}^{\prime}}$}
  {
  $\eta_{1}^{\star} \gets \eta_{1} - \delta_{\eta_{1}}$, $\eta_{2}^{\star} \gets \eta_{2}^{\ast}$, $\alpha^{\star} \gets \alpha^{\ast}$; break
  }
  \ElseIf{$\left\vert P_{e}^{\iota}-P_{e}^{o}\right\vert < \delta_{P_{e}^{\prime}}$}
  {
  $\eta_{1}^{\star} \gets \eta_{1}$, $\eta_{2}^{\star} \gets \eta_{2}^{i}$, $\alpha^{\star} \gets \alpha^{i}$; break\\
  }
  }
  \caption{\label{Algo:M2} Two-Layer Greedy Descent Algorithm}
\end{algorithm}

\begin{figure}[!htb]
    \centering
    \begin{minipage}[t]{.32\textwidth}
        \centering
        \includegraphics[width = \textwidth, height = 0.9\textwidth]{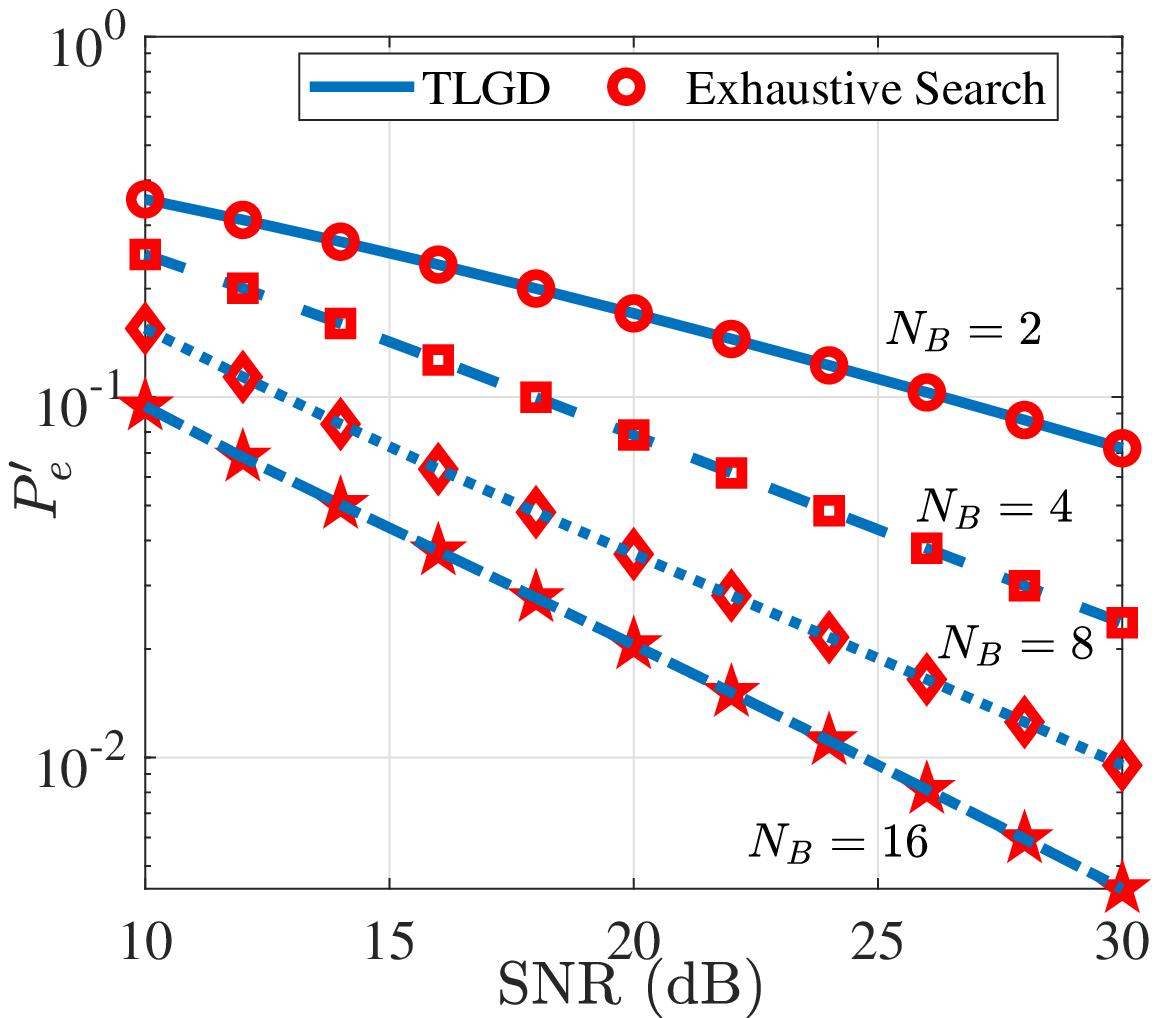}
        \caption{\label{fig:Joint_per} Performance of NC-FFFD using energy levels obtained using TLGD and the exhaustive search.}
    \end{minipage}%
        \hfill
    \begin{minipage}[t]{0.32\textwidth}
        \centering
        \includegraphics[width = \textwidth, height = 0.9\textwidth]{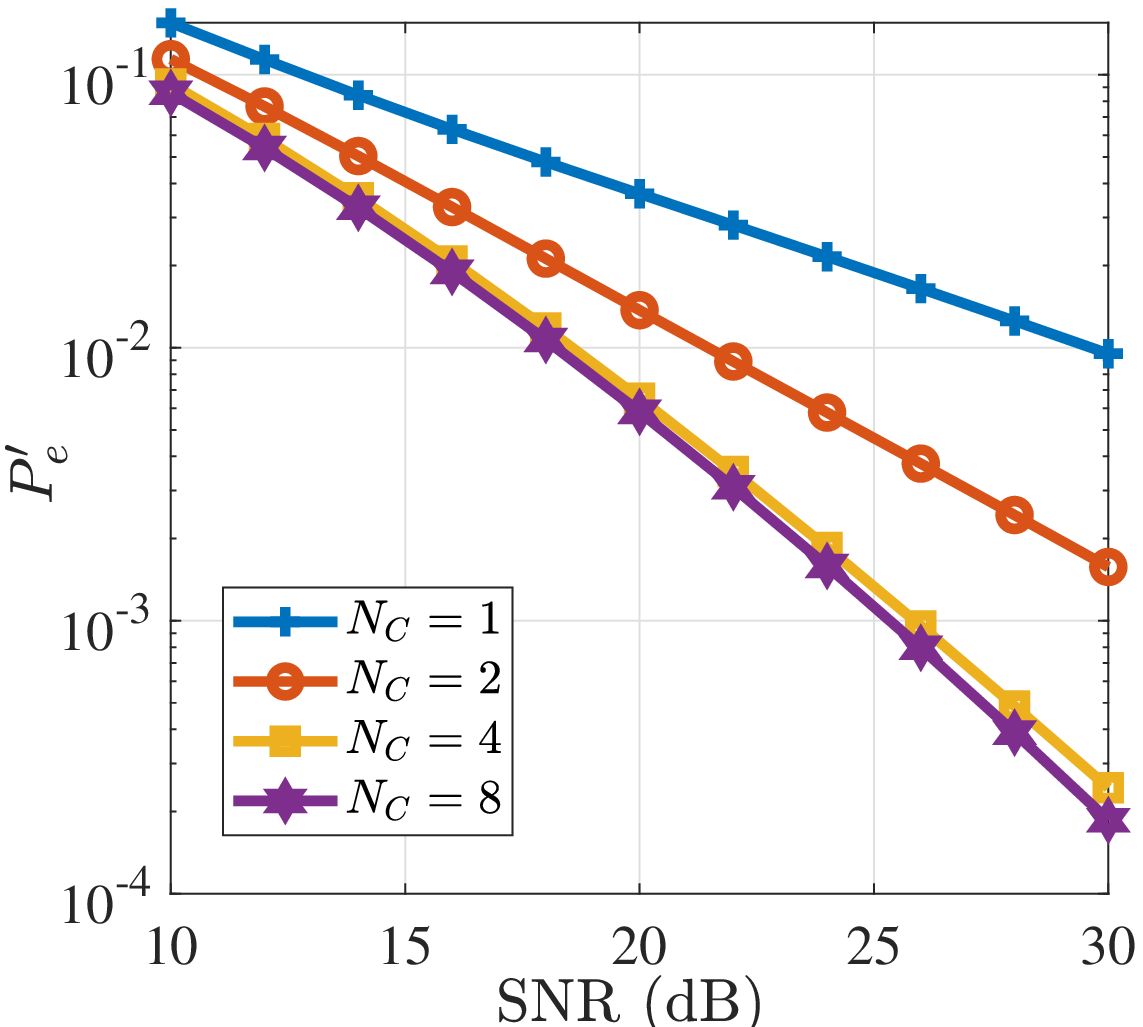}
\caption{\label{fig:Pe_OOK_varNc} Performance of NC-FFFD for fixed $N_{B}=8$ and varying $N_{C}$.}
    \end{minipage}%
    \hfill
    \begin{minipage}[t]{0.32\textwidth}
        \centering
        \includegraphics[width = \textwidth, height = 0.9\textwidth]{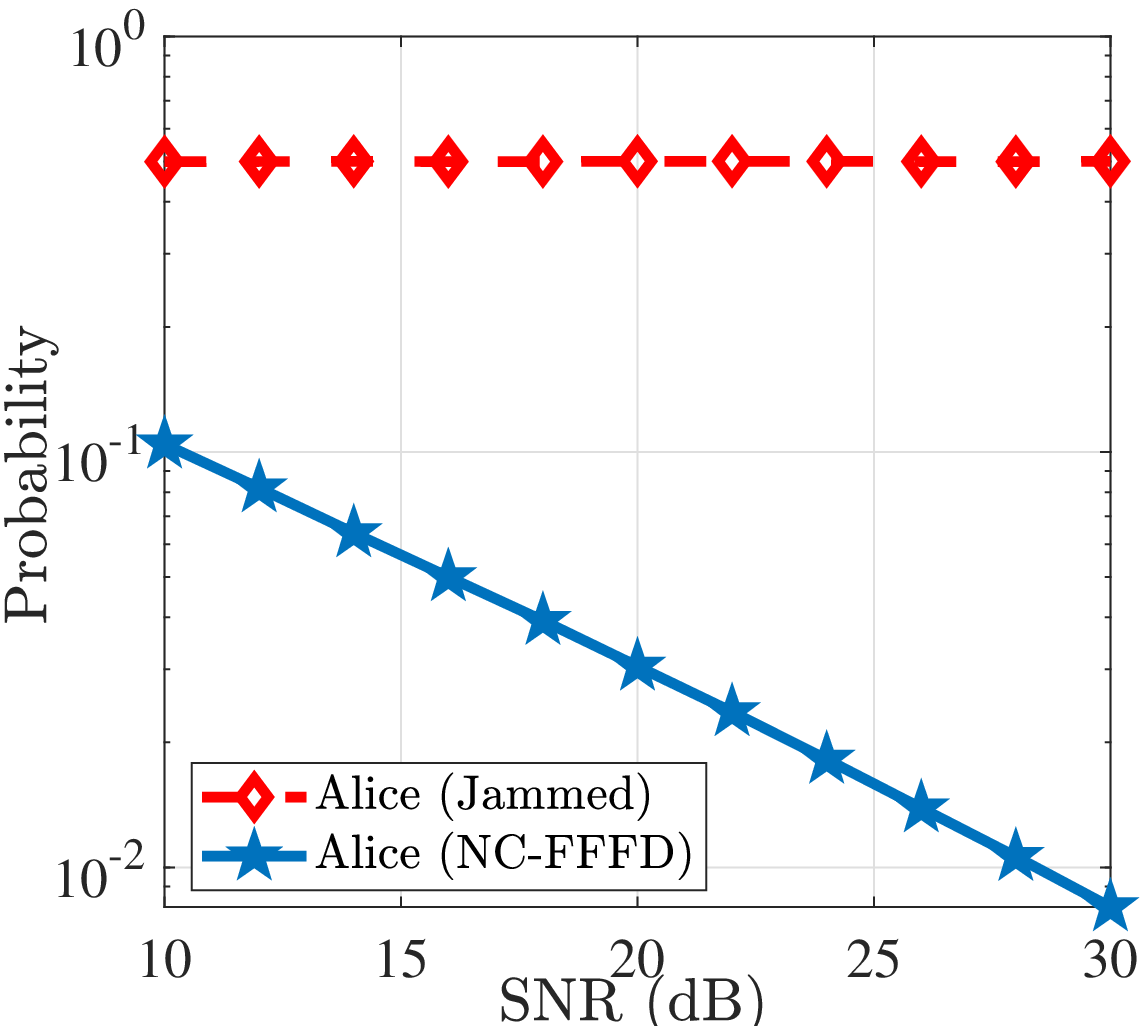}
\caption{\label{fig:Alice_per} Alice's performance when using NC-FFFD scheme for $N_{C}=1$ and $N_{B}=8$.}
    \end{minipage}     
\end{figure} 

In Fig.~\ref{fig:Joint_per}, we plot the error performance of NC-FFFD scheme as a function of SNR and $N_{B}$ using Monte-Carlo simulations. We assume, $\sigma_{AB}^{2}=\sigma_{CB}^{2}=1$, $\lambda=-50$ dB, and $N_{C}=1$. Further, due to vicinity of Alice and Charlie, we assume $\sigma_{AC}^{2}=4$, thus, providing $6$ dB improvement in SNR on Alice-to-Charlie link as compared to Alice-to-Bob link. We compute the error-rates when the optimal energy levels and $\alpha$ are obtained using exhaustive search on  \eqref{eq:Pe_M2}. We also compute the error-rates using the proposed algorithm. For both the scenarios, we observe that the error curves approximately overlap, indicating the efficacy of the proposed algorithm, as well as our approach of using \eqref{opt:M2} instead of \eqref{eq:Pe_M2}. Further, in Fig~\ref{fig:Pe_OOK_varNc}, for same parameters and $N_{B}=8$, we plot the  error performance of ND-FFFD scheme as a function of SNR for various values of $N_{C}$. We observe that, the error performance of NC-FFFD scheme improves as a function of $N_{C}$. Finally, for the same parameters and $N_{B}=8$, in Fig.~\ref{fig:Alice_per}, we show the improvement in Alice's performance when using NC-FFFD relaying scheme. In terms of feasibility of implementation, the complexity analysis of TLGD algorithm has been discussed in the conference proceedings of this work \cite{my_GCOM}. 

\subsection{Optimization of Energy Levels for $M\geq 2$}
\label{ssec:gncfffd}
In this section, we provide a solution that computes the optimal energy levels, $\{\epsilon_{j},\eta_{j}\}$, and the factor $\alpha$, when $M\geq 2$. Since the average transmit energy of Charlie is constrained to $\mathrm{E}_{C,f_{CB}}$, increasing the data-rate at Charlie results in degraded joint error performance as compared to $M=2$. One way to improve the error performance is by using a large number of receive antennas at Bob. Despite this improvement, it is important to note that the joint error performance is also a function of the SNR of Alice-to-Charlie link. Therefore, an improved Alice-to-Charlie link can help to improve the overall performance of the scheme. This is also evident from Fig.~\ref{fig:Pe_OOK_varNc}, where we observe that the error performance of the scheme improves as a function of $N_{C}$. This motivates us to solve $P_{e}$ in \eqref{opt} for optimal $\{\epsilon_{j},\eta_{j}\}$, and $\alpha$ under the assumption that Charlie has a sufficiently large number of receive-antennas. In this section, we take a similar approach as that of Sec.~\ref{ssec:Globecom}, by upper bounding the complementary error terms by $1$ to obtain an upper bound on $P_{e}$ given by,

\begin{small}
\bieee
P_{e}\leq P_{e}^{\prime} = \frac{1}{2M}\left[ \sum_{\ell_{1} = 1}^{M}P_{00}P_{e, S_{\frac{1}{2}\left(4\ell_{1}+(-1)^{\ell_{1}}-1\right)}} + \sum_{\ell_{2} = 1}^{M}P_{11}P_{e, S_{\frac{1}{2}\left((-1)^{\ell_{2}}\left(4(-1)^{\ell_{2}}\ell_{2} + (-1)^{\ell_{2}+1}-1\right)\right)}} + M\left(P_{01}+P_{10}\right)\right].\label{eq:Pe_upper} 
\eieee
\end{small}
\noindent Since $P_{e}^{\prime}$ is a function of $S_{\ell}$ and $\alpha$, besides $N_{C}$, $N_{B}$, and SNR, in the next theorem, we compute the optimal value of $\alpha\in(0,1)$, that minimises $P_{e}^{\prime}$, when $S_{1},\cdots,S_{2M}$, $N_{C}$, $N_{B}$, and SNR are fixed.
\begin{theorem}
\label{th:alpha_range}
When $S_{1},\cdots,S_{2M}$ are fixed, such that $S_{2}<1$, the optimal value of $\alpha\in(0,1)$ that minimises $P_{e}^{\prime}$ in \eqref{eq:Pe_upper} is given by, $\alpha^{\dagger} = 1-S_{2}$. 
\end{theorem}
\begin{proof}
We will first show that $P_{e}^{\prime}$ in \eqref{eq:Pe_upper} is an increasing function of $\alpha$. Then, we compute a lower bound on $\alpha$ considering the feasible energy levels jointly contributed by Alice and Charlie.

The expression of $P_{e}^{\prime}$ in \eqref{eq:Pe_upper} is a convex combination of $P_{00}$, $P_{01}$, $P_{10}$, and $P_{11}$. Further, we notice that $P_{00}$ and $P_{11}$ are decreasing functions of $\alpha$ (Lemma~\ref{lm:P10P01_alpha}). However, since $S_{1},\cdots,S_{2M}$ are fixed, the coefficients of $P_{00}$ and $P_{11}$ are independent of $\alpha$, such that, $\sum_{\ell_{1} = 1}^{M}P_{e, S_{\frac{1}{2}\left(4\ell_{1}+(-1)^{\ell_{1}}-1\right)}}\leq  M$ and $\sum_{\ell_{2} = 1}^{M}P_{e, S_{\frac{1}{2}\left((-1)^{\ell_{2}}\left(4(-1)^{\ell_{2}}\ell_{2} + (-1)^{\ell_{2}+1}-1\right)\right)}}\leq  M$. Further, since $P_{01}$ and $P_{10}$ are increasing functions of $\alpha$, it is straightforward that, $P_{e}^{\prime}$ is an increasing function of $\alpha$. This completes the first part of the proof.

Although, we upper bound the energy levels $\overline{S}_{\ell}$ by $1$, in practice, Bob receives these energy levels when $e=1$ at Charlie. From \eqref{eq:map2}, we have, $\overline{S}_{\frac{1}{2}\left(4\ell_{1} + (-1)^{\ell_{1}}-1\right)} = S_{\frac{1}{2}\left((-1)^{\ell_{1}}\left(4(-1)^{\ell_{1}}+(-1)^{\ell_{1}+1}-1\right)\right)}-(1-\alpha)$. It is important to note that, if $S_{\frac{1}{2}\left((-1)^{\ell_{1}}\left(4(-1)^{\ell_{1}}+(-1)^{\ell_{1}+1}-1\right)\right)}<1-\alpha$ , then $\overline{S}_{\frac{1}{2}\left(4\ell_{1} + (-1)^{\ell_{1}}-1\right)}<0$. However, since $\overline{S}_{\ell}\in\mathcal{\overline{S}}$ are energy levels, $\overline{S}_{\ell}\geq 0$. Therefore, to achieve $\overline{S}_{\frac{1}{2}\left(4\ell_{1} + (-1)^{\ell_{1}}-1\right)}\geq 0$,  we must have $S_{\frac{1}{2}\left((-1)^{\ell_{1}}\left(4(-1)^{\ell_{1}}+(-1)^{\ell_{1}+1}-1\right)\right)}\geq 1-\alpha$ or $\alpha\geq 1-S_{\frac{1}{2}\left((-1)^{\ell_{1}}\left(4(-1)^{\ell_{1}}+(-1)^{\ell_{1}+1}-1\right)\right)}$. Therefore, $\alpha\geq\max\left\{1-S_{\frac{1}{2}\left((-1)^{\ell_{1}}\left(4(-1)^{\ell_{1}}+(-1)^{\ell_{1}+1}-1\right)\right)}\right\}$, where $\ell_{1}=1,\cdots,M$. However, we know that, $S_{1}<\cdots<S_{2M}$, thus, we have $\alpha\geq 1-S_{2}$.

Finally, since $P_{e}^{\prime}$ in \eqref{eq:Pe_upper} is an increasing function of $\alpha$ and $\alpha\geq 1-S_{2}$, $P_{e}^{\prime}$ is minimised when $\alpha=\alpha^{\dagger}=1-S_{2}$. 
\end{proof}

The result of Lemma~\ref{lm:P10P01_nc} indicates that $P_{01}$ and $P_{10}$ are decreasing functions of $N_{C}$. Further, $S_{\ell}$, $\ell=1,\cdots,2M$ are independent of $N_{C}$, as a result, each convex combination in \eqref{eq:Pe_upper} decreases as $N_{C}$ increases. Therefore, it is straightforward to prove that $P_{e}^{\prime}$ is a decreasing function of $N_{C}$. 
\begin{proposition}
\label{prop:Pe_nc_dec}
For a fixed $\alpha\in(0,1)$, when $N_{C}\rightarrow\infty$, we have $P_{01}=P_{10}\approx 0$ and $P_{00}=P_{11}\approx 1$, we have, $P_{e}^{\prime}\geq P_{e,approx} = \frac{1}{2M}\!\left[ \sum_{\ell_{1} = 1}^{M}P_{e, S_{\frac{1}{2}\left(4\ell_{1}+(-1)^{\ell_{1}}-1\right)}} + \sum_{\ell_{2} = 1}^{M}P_{e, S_{\frac{1}{2}\left((-1)^{\ell_{2}}\left(4(-1)^{\ell_{2}}\ell_{2} + (-1)^{\ell_{2}+1}-1\right)\right)}}\right]$.
\end{proposition} 

Motivated by the result of Proposition~\ref{prop:Pe_nc_dec}, instead of solving \eqref{opt} for a sufficiently large $N_{C}$ using the first principles, we take an alternate approach, where we first compute $S_{1},\cdots,S_{2M}$ that minimises $P_{e,approx}$ and then compute the respective $\{\epsilon_{j},\eta_{j}\}$, and $\alpha$ using the relation in \eqref{eq:map2}.

Towards computing the optimal $S_{1},\cdots,S_{2M}$, we observe that since an energy level $S_{\ell}$ corresponds to the sum of energies contributed by Alice, Charlie, and the AWGN at Bob on $f_{CB}$, the sum energies contributed by Alice and Charlie will be $S_{\ell}-N_{o}$. Furthermore, since the average energy on $f_{CB}$ is $1$, we have the following constraint of $S_{\ell}$:
\bieee
\dfrac{1}{2M}\sum_{\ell=1}^{2M}\left(S_{\ell} - N_{o}\right) = 1.\label{eq:sum_const}
\eieee
Finally, we formulate the following optimization problem of computing optimal $S_{1}, \cdots, S_{2M}$ so as to minimise $P_{e,approx}$, subject to \eqref{eq:sum_const}.
\bieee
S_{1}^{\star},\cdots,S_{2M}^{\star} = \arg\underset{S_{1},\cdots,S_{2M}}{\min} \quad &  & P_{e,approx}\label{opt2}\\
 \text{subject to:} \quad & &\dfrac{1}{2M}\sum_{\ell=1}^{2M}\left(S_{\ell} - N_{o}\right) = 1, S_{1}<\cdots < S_{2M}.\nn
\eieee 

While \eqref{opt2} can be solved using the first principles, \cite{ranjan} provides a near-optimal solution for \eqref{opt2}. Therefore, we use the results of \cite{ranjan} to compute $S_{1}^{\star},\cdots,S_{2M}^{\star}$. In the next lemma, we prove that, when we use $S_{1},\cdots,S_{2M}$ to obtain $\{\epsilon_{j},\eta_{j}\}$, such that $S_{1},\cdots, S_{2M}$ follows \eqref{eq:sum_const}, $\{\epsilon_{j},\eta_{j}\}$ satisfies \eqref{eq:new_constaint}.

\begin{lemma}
If $S_{1},\cdots,S_{2M}$ are fixed such that \eqref{eq:sum_const} is satisfied, then the average transmit energy of Charlie is given by \eqref{eq:new_constaint}.
\end{lemma}
\begin{proof}
From \eqref{eq:map2}, we have $S_{\frac{1}{2}\left(4\ell_{1} + (-1)^{\ell_{1}}-1\right)} = \epsilon_{\ell_{1}}+N_{o},$ and $S_{\frac{1}{2}\left((-1)^{\ell_{1}}\left(4(-1)^{\ell_{1}}+(-1)^{\ell_{1}+1}-1\right)\right)} = 1-\alpha + \eta_{\ell_{1}} + N_{o}$ for $i=0,1$, respectively, where $\ell_{1}=1,\cdots,M$. Rearranging and summing LHS and RHS of both the equations, we get, $\sum_{\ell=1}^{2M}(S_{\ell} - N_{o}) = \sum_{\ell_{1}=1}^{M}\left(\epsilon_{\ell_{1}}+\eta_{\ell_{1}} + (1-\alpha)\right)$. Dividing both sides by $2M$ and rearranging, we get \eqref{eq:new_constaint}.
\end{proof}
In the next section, we propose the energy backtracking algorithm, where we first solve \eqref{opt2} using \cite{ranjan} to obtain $S_{1}^{\star},\cdots,S_{2M}^{\star}$ and then compute corresponding $\{\epsilon_{j},\eta_{j}\vert j=1,\cdots,M\}$, and $\alpha$.  It is important to note that, since Charlie cannot have $N_{C}\rightarrow\infty$, we must bound the number of receive-antennas at Charlie. Thus, we use a parameter $0<\Delta_{RE}\ll 1$ to bound $N_{C}$. Therefore, we compute the minimum number of receive-antennas at Charlie, such that the relative error between $P_{e,approx}^{\star}$ and $P_{e,eval}$ is within $\Delta_{RE}$, where $P_{e,approx}^{\star}$ is $P_{e,approx}$ evaluated at $S_{1}^{\star},\cdots,S_{2M}^{\star}$ and $P_{e,eval}$ is $P_{e}$ evaluated at optimal $\{\epsilon_{j},\eta_{j}\vert j=1,\cdots,M\}$, and $\alpha$.

\subsection{Energy Backtracking (EB) Algorithm}
The Energy Backtracking (EB) Algorithm, first computes energy levels $S_{1}^{\star},\cdots,S_{2M}^{\star}$ using the semi-analytical results of \cite{ranjan}. It then computes $\alpha^{\dagger}$, and $\epsilon_{j}^{\dagger}$ and $\eta_{j}^{\dagger}$ based on Theorem~\ref{th:alpha_range} and the relation in \eqref{eq:map2}, respectively. It then sets $N_{C}=1$ and computes $P_{e,eval}$, i.e., $P_{e}$ at $\alpha^{\dagger}$, $\epsilon_{j}^{\dagger}$, $\eta_{j}^{\dagger}$ for the given $N_{B}$. The algorithm increments $N_{C}$ until relative error between the $P_{e,approx}^{\star}$ and $P_{e,eval}$ is within $\Delta_{RE}$. The algorithm exits the while-loop when the relative error is less than or equal to $\Delta_{RE}$. The pseudo-code for the proposed  EB algorithm is given in Algorithm~\ref{Algo:Generalised}. 
\begin{algorithm}
\setstretch{0.32}
\DontPrintSemicolon 
  \KwInput{$P_{e}$ \eqref{eq:Pe}, $P_{e,approx}$, $\Delta_{RE}>0$, $M$, $N_{B}$, $N_{o}$}
  \KwOutput{$\epsilon_{1}^{\dagger},\cdots,\epsilon_{M}^{\dagger}$, $\eta_{1}^{\dagger},\cdots,\eta_{M}^{\dagger}$, $N_{C}^{\dagger}$, $\alpha^{\dagger}$}
Compute $S_{1}^{\star},\cdots,S_{2M}^{\star}$ using \cite{ranjan} and evaluate $P_{e,approx}^{\star}$.\\
$\alpha^{\dagger} = 1-S_{2}^{\star}$\\
  ; $\epsilon_{j}^{\dagger} = S_{\frac{1}{2}\left(4j + (-1)^{j}-1\right)}^{\star}-N_{o}$; $\eta_{j}^{\dagger} = S_{\frac{1}{2}\left((-1)^{j}\left(4(-1)^{j}+(-1)^{j+1}-1\right)\right)}^{\star} - (1-\alpha^{\dagger})-N_{o}, \ j=1,\cdots,M$\\
 Set: $N_{C}=1$, $P_{e,eval}=1$\\
  \While{$\left\vert\dfrac{P_{e,approx}^{\star}-P_{e,eval}}{P_{e,approx}^{\star}}\right\vert\geq\Delta_{RE}$} 
  {
  Substitute $S_{1}^{\star},\cdots,S_{2M}^{\star}$, $\alpha^{\dagger}$, $N_{C}$, and $N_{B}$ in \eqref{eq:Pe} and obtain $P_{e,eval}$

   \If{$\left\vert\dfrac{P_{e,approx}^{\star}-P_{e,eval}}{P_{e,approx}^{\star}}\right\vert >\Delta_{RE}$}
  {
  $N_{C}=N_{C}+1$; continue
  }
  \Else
  {
  $N_{C}^{\dagger}=N_{C}$; break
  }
  }
  \caption{\label{Algo:Generalised} Energy Backtracking Algorithm}
\end{algorithm}

\begin{figure}[t]
\vspace{-0.15in}
    \centering
    \begin{minipage}[t]{0.32\textwidth}
       \centering
       \includegraphics[width = \textwidth, height = 0.9\textwidth]{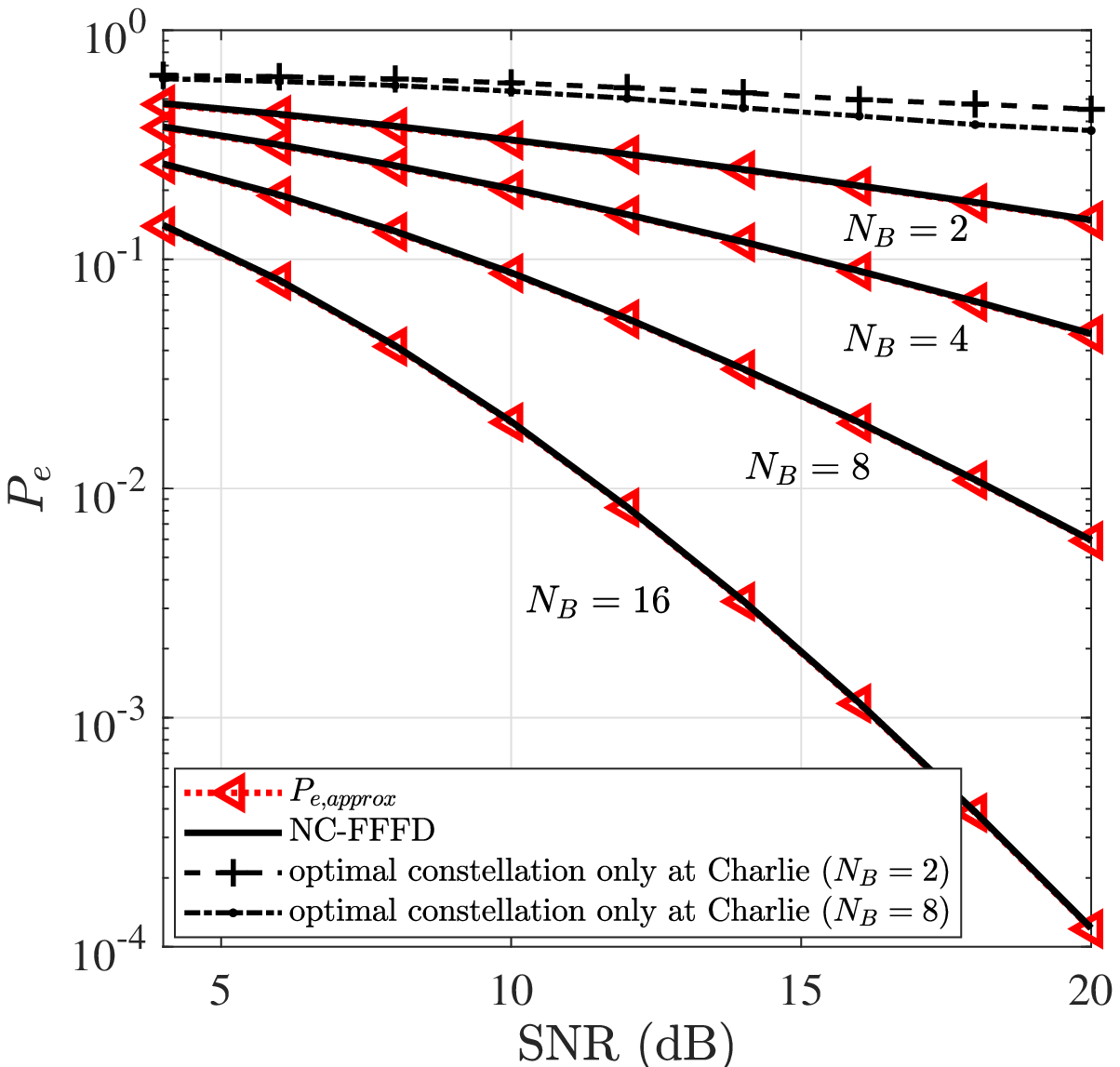}
        \caption{\label{fig:Algo2_M2} Error performance of NC-FFFD when energy levels are computed using EB algorithm for $M=2$.}
     \end{minipage}%
        \hfill
    \begin{minipage}[t]{0.32\textwidth}
        \centering
       \includegraphics[width = \textwidth, height = 0.9\textwidth]{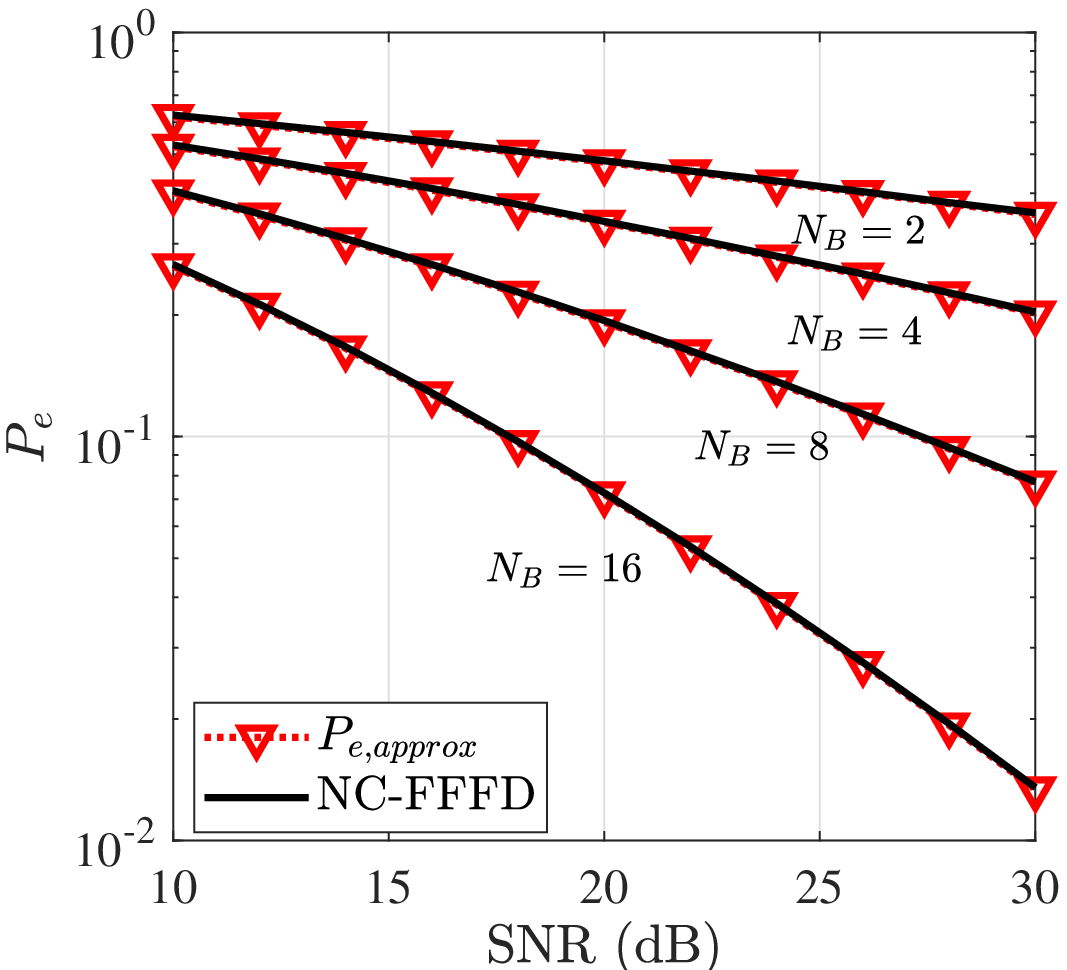}
\caption{\label{fig:Algo2_M4} Error performance of NC-FFFD when energy levels are computed using EB algorithm for $M=4$.}
    \end{minipage}%
    \hfill
    \begin{minipage}[t]{0.32\textwidth}
        \centering
       \includegraphics[width = \textwidth, height = 0.9\textwidth]{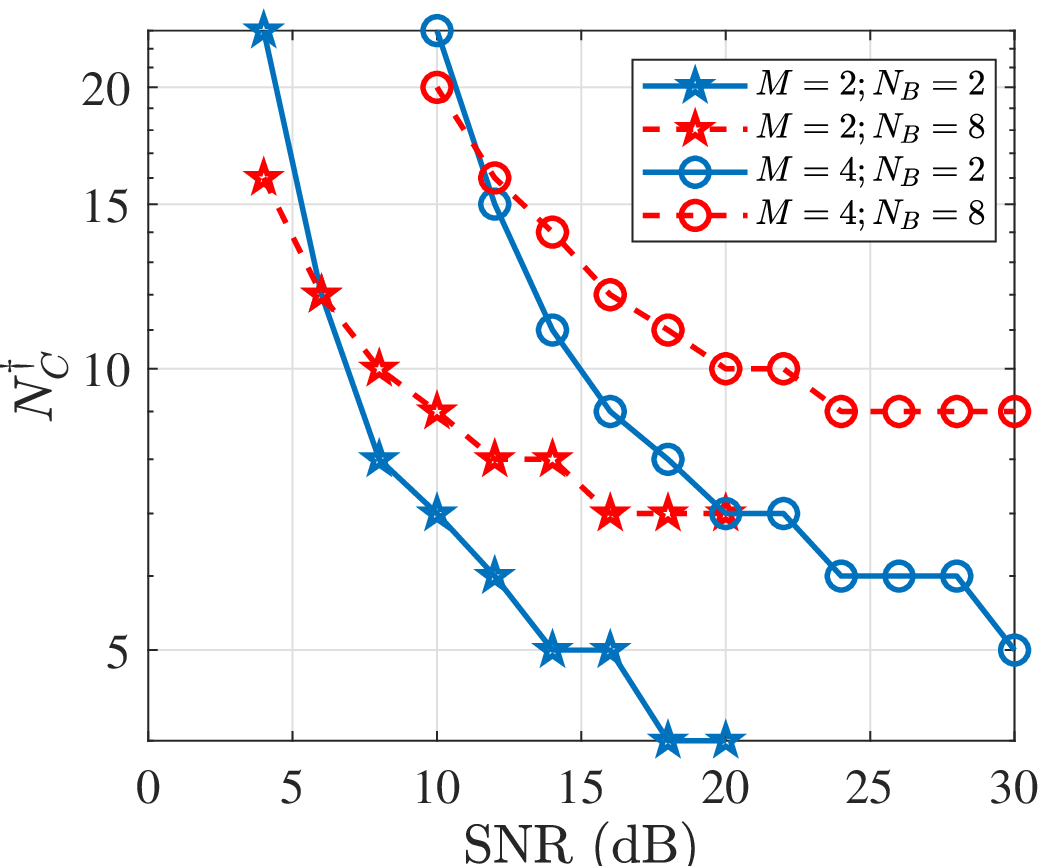}
\caption{\label{fig:opt_ant}$N_{C}^{\dagger}$ as a function of SNR for $M=2$ and $M=4$.}
    \end{minipage}
\end{figure}

In Fig.~\ref{fig:Algo2_M2} and Fig.~\ref{fig:Algo2_M4}, we plot the error performance of NC-FFFD relaying scheme when using the EB Algorithm for $M=2$ and $M=4$ for various values of $N_{B}$. In addition to the simulation parameters assumed above, we assume $\Delta_{RE}=10^{-2}$ for both the cases. For both the cases, we observe that the error performance improves as a function of SNR. In Fig.~\ref{fig:Algo2_M2}, we also plot the performance of NC-FFFD scheme when Charlie uses optimal energy levels for point-to-point communication obtained using \cite{ranjan} for $N_{B}=2,8$. From the plot it is clear that, although Charlie is using optimal energy levels for point-to-point communication, the corresponding error performance of NC-FFFD is poor. This reinforces that to minimise $P_{e}$, energy levels at both the users must be jointly optimised as proposed in Algorithm~\ref{Algo:Generalised}. Finally, in Fig.~\ref{fig:opt_ant}, we also plot $N_{C}^{\dagger}$ as a function of SNR, for various combinations of $M$ and $N_{B}$ and observe that with only tens of antennas at the helper, we can achieve the performance close to its large-antenna counterpart.

If NC-FFFD scheme provides performance close to $P_{e,approx}^{\star}$, it assumes that fast-forwarding at Charlie is perfect. Therefore, the symbols on the direct link, i.e., Alice-to-Bob link and relayed link, i.e., Charlie-to-Bob link, arrive during the same time instant, thereby resulting in the signal model in \eqref{eq:rb}. In the next section, we discuss the case when fast-forwarding at Charlie is imperfect. In particular, we discuss the consequences related to this problem and a possible solution.

\section{Delay-Tolerant NC-FFFD (DT NC-FFFD) Relaying Scheme}
\label{sec:DT_NC-FFFD}
If $nT$ denotes the delay on the relayed link w.r.t. the direct link, such that $n\geq 0$ and $T$ is the symbol duration, then $nT=0$, when fast-forwarding is perfect. However, when fast-forwarding is imperfect, $nT\neq 0$ and $\mathbf{r}_{B}$ must be a function of $nT$. In particular, when $nT\neq 0$, the symbol received at Bob is a function of Alice's current symbol, Charlie's current symbol, and Alice's symbol delayed by $nT$. Although, Charlie's current symbol and Alice's symbol delayed by $nT$ are captured by $E_{C}$, the current symbol of Alice creates an interference in the symbol decoding, thereby degrading the error performance. To illustrate this behaviour, we plot the error performance of NC-FFFD scheme in Fig.~\ref{fig:DT1}, when the symbols on the relayed link arrive one symbol period after the symbols on the direct link. The plot shows that, the error performance degrades as the energy on the direct link interferes when Bob tries to decode symbols using the relayed link. 
 
Towards computing the optimal energy levels at Alice and Charlie when $nT\neq 0$, one can formulate a new signal model, where $\mathbf{r}_{B}$ is a function of $nT$ and then compute the optimal energy levels using the first principles. However, we note that, Alice contributes \emph{zero} and $1-\alpha$ energies on the direct link, when she transmits symbol $0$ and symbol $1$, respectively. Thus, in order to reduce the interference from the direct link, we must reduce the term $1-\alpha$. Therefore, if we upper bound the contribution $1-\alpha$ by small value, then we can continue to use the same signal model on $\mathbf{r}_{B}$ as given in \eqref{eq:rb}, thereby making NC-FFFD scheme \emph{Delay Tolerant}. To this end, we propose an upper bound on $1-\alpha$ as, $1-\alpha\leq \Delta_{\text{DT}}N_{o}$, where $0<\Delta_{\text{DT}}\ll 1$ is the design parameter. Since $1-\alpha\leq \Delta_{\text{DT}}N_{o}$, we have the relation $\alpha\geq 1-\Delta_{\text{DT}}N_{o}$. Further, the result of Theorem~\ref{th:alpha_range} shows that $P_{e}^{\prime}$ is an increasing function of $\alpha$, therefore, the optimal choice of $\alpha$ would be, $\alpha= 1-\Delta_{\text{DT}}N_{o}$. However, since $\Delta_{\text{DT}}\ll 1$, $1-S_{2}<1-\Delta_{\text{DT}}N_{o}$ and therefore, using $\alpha=1-\Delta_{\text{DT}}N_{o}$ will degrade the error performance. In the following discussion, we show that we can achieve the same error performance at $\alpha = 1-\Delta_{\text{DT}}N_{o}$ as achieved in Sec.~\ref{ssec:gncfffd} at $\alpha=1-S_{2}$, by increasing the receive-diversity at Charlie.
\begin{figure}[!htb]
\begin{center}
    \begin{minipage}[t]{0.32\textwidth}
       \centering
      \includegraphics[width = \textwidth, height = 0.9\textwidth]{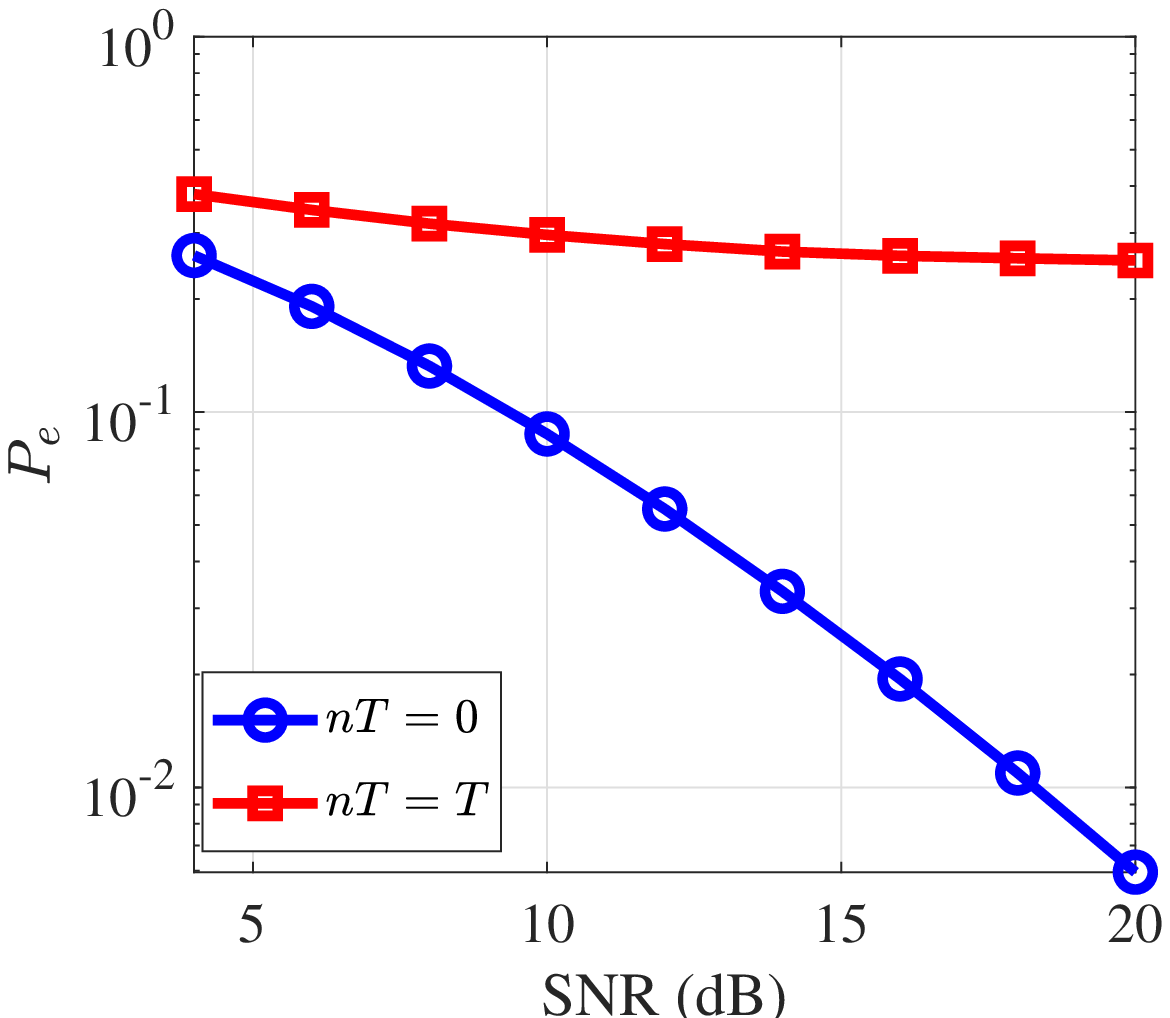}
        \caption{\label{fig:DT1} Error performance for $nT=0$ and $nT=T$.}
     \end{minipage}%
        \hfill
    \begin{minipage}[t]{0.32\textwidth}
        \centering
        \includegraphics[width = \textwidth, height = 0.9\textwidth]{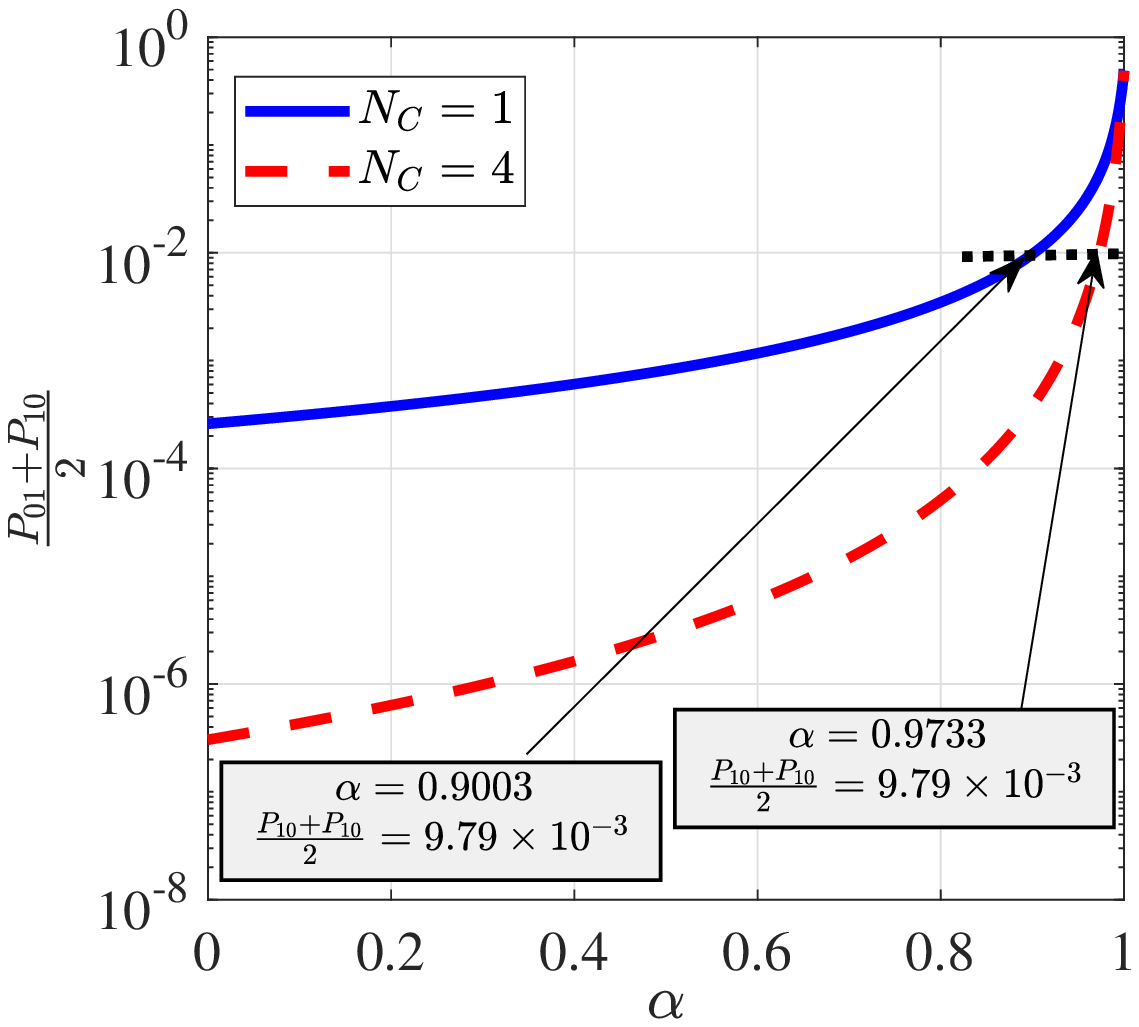}
\caption{\label{fig:DT2} Variation of $\frac{P_{01}+P_{10}}{2}$ as a function of $N_{C}$ and $\alpha$.}
    \end{minipage}%
    \hfill
    \begin{minipage}[t]{0.32\textwidth}
        \centering
     \includegraphics[width = \textwidth, height = 0.9\textwidth]{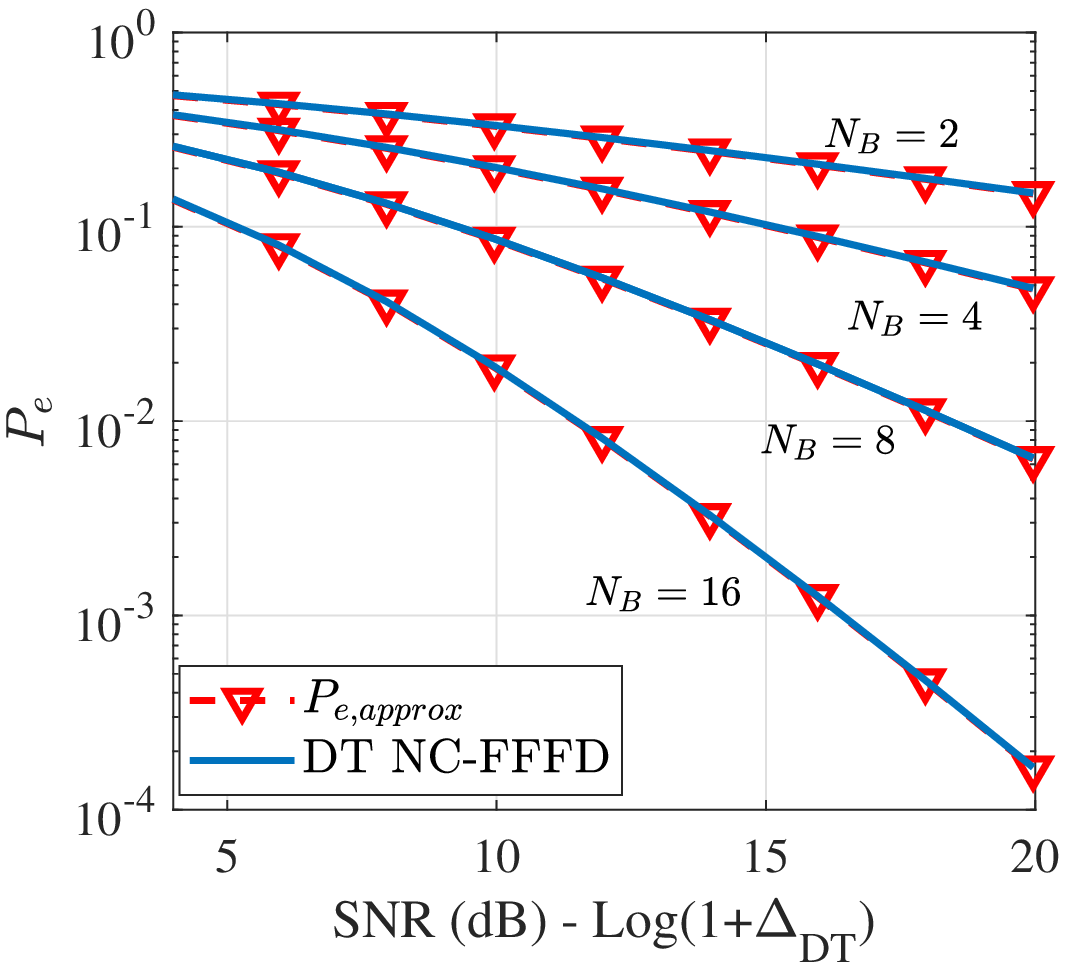}
        \caption{\label{fig:DT_M2}Performance of DT NC-FFFD when energy levels are computed using DT-EB algorithm for $M=2$.}
    \end{minipage}
    \end{center}
\end{figure}

Since the terms $P_{00}$, $P_{01}$, $P_{10}$, and $P_{11}$ are  functions of $\alpha$ and $N_{C}$ in $P_{e}^{\prime}$, we show that  one can achieve the same $P_{00}$, $P_{01}$, $P_{10}$, and $P_{11}$ at different combinations of $\alpha$ and $N_{C}$. The results of Lemma~\ref{lm:P10P01_alpha} show that for a fixed $N_{C}$, $P_{01}$ and $P_{10}$ are increasing functions of $\alpha$. Subsequently, from Lemma~\ref{lm:P10P01_nc}, for a fixed $\alpha$, $P_{01}$ and $P_{10}$ are decreasing functions of $N_{C}$. In Fig.~\ref{fig:DT2}, we plot $\frac{P_{01}+P_{10}}{2}$ as a function of $\alpha$ for various $N_{C}$ at $25$ dB and observe that, for $N_{C}=1$ and $\alpha = 0.9003$, the average probability of error of Alice-to-Charlie link is $9.79\times 10^{-3}$. However, to obtain the same error performance at larger $\alpha$, i.e., $\alpha=0.9733$, we must use $N_{C}=4$.  

Based on the above discussion, in the next section, we propose a variant of EB algorithm, where we bound the interference from the direct link by $\Delta_{\text{DT}}N_{o}$ and obtain $\{\epsilon_{j},\eta_{j}\}$ and the minimum $N_{C}$, such that the error performance is close to $P_{e,approx}$.

\subsection{Delay Tolerant Energy Backtracking (DT-EB) Algorithm}
In the Delay Tolerant Energy Backtracking (DT-EB) algorithm, we obtain the optimal energy levels at Alice and Charlie, such that the energy level on the direct link is bounded by $\Delta_{\text{DT}}N_{o}$. To facilitate this, we use the EB algorithm with two variations, i) we set $\alpha=1-\Delta_{\text{DT}}N_{o}$, instead of $\alpha = 1-S_{2}^{\star}$, ii) the effective SNR to compute $S_{1}^{\star},\cdots,S_{2M}^{\star}$ is $\left(N_{o}+\Delta_{\text{DT}}N_{o}\right)^{-1}$. 
\begin{figure}[!htb]
\begin{center}
    \begin{minipage}[t]{0.32\textwidth}
        \centering
        \includegraphics[width = \textwidth, height = 0.9\textwidth]{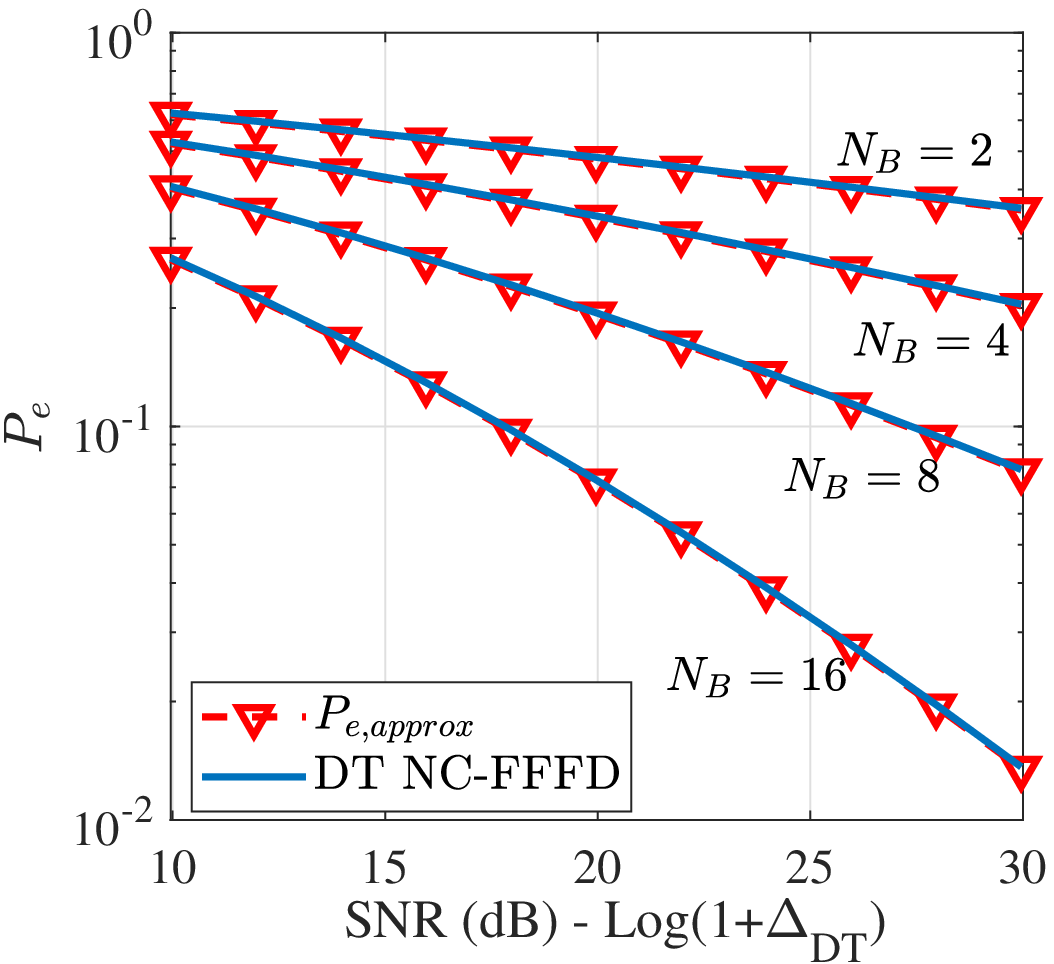}
\caption{\label{fig:DT_M4}Performance of DT NC-FFFD when energy levels are computed using DT-EB algorithm for $M=4$.}
    \end{minipage}%
    \hfill
    \begin{minipage}[t]{0.32\textwidth}
       \centering
      \includegraphics[width = \textwidth, height = 0.9\textwidth]{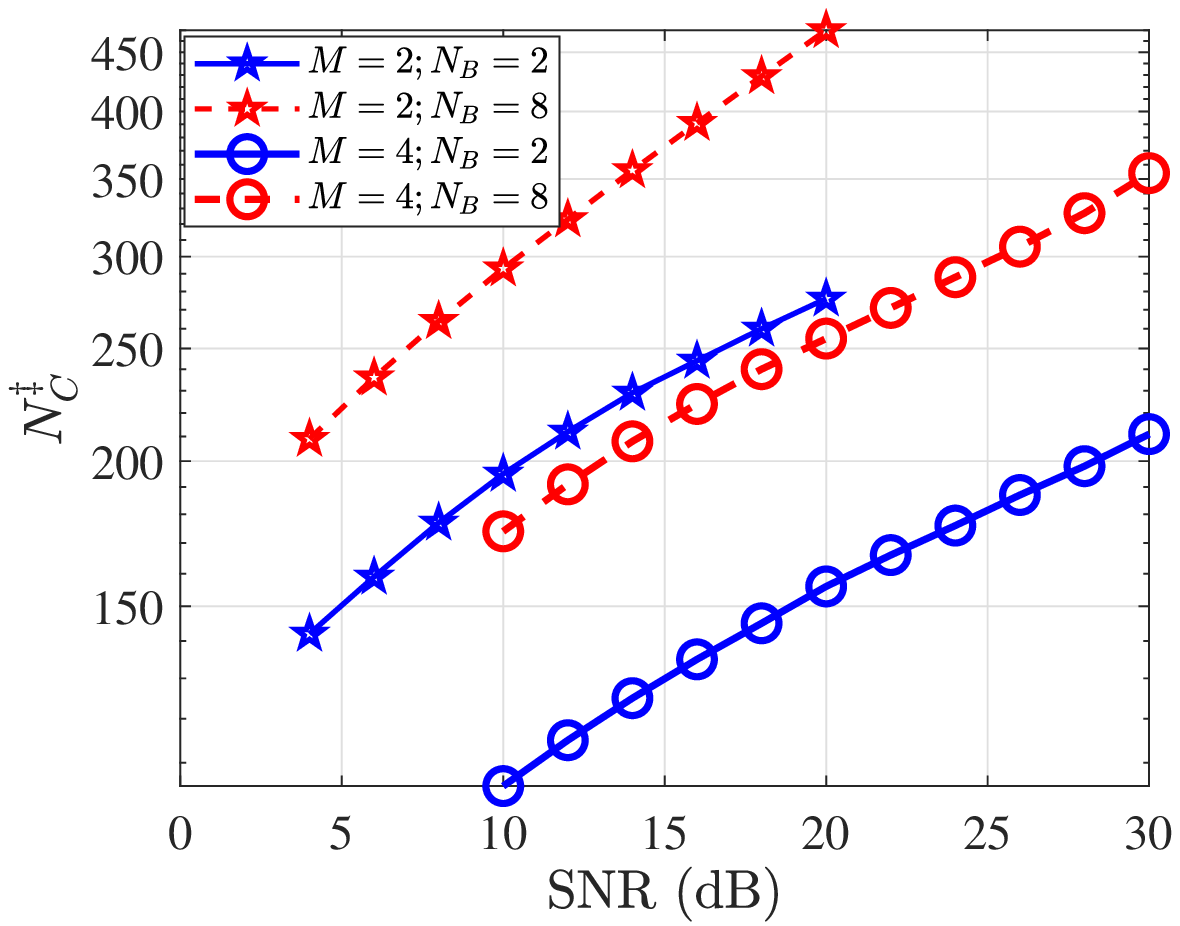}
        \caption{\label{fig:opt_ant_DT} $N_{C}^{\ddagger}$ as a function of SNR for $M=2$ and $M=4$.}
     \end{minipage}%
     \hfill
    \begin{minipage}[t]{0.32\textwidth}
        \centering
      \includegraphics[width = \textwidth, height = 0.9\textwidth]{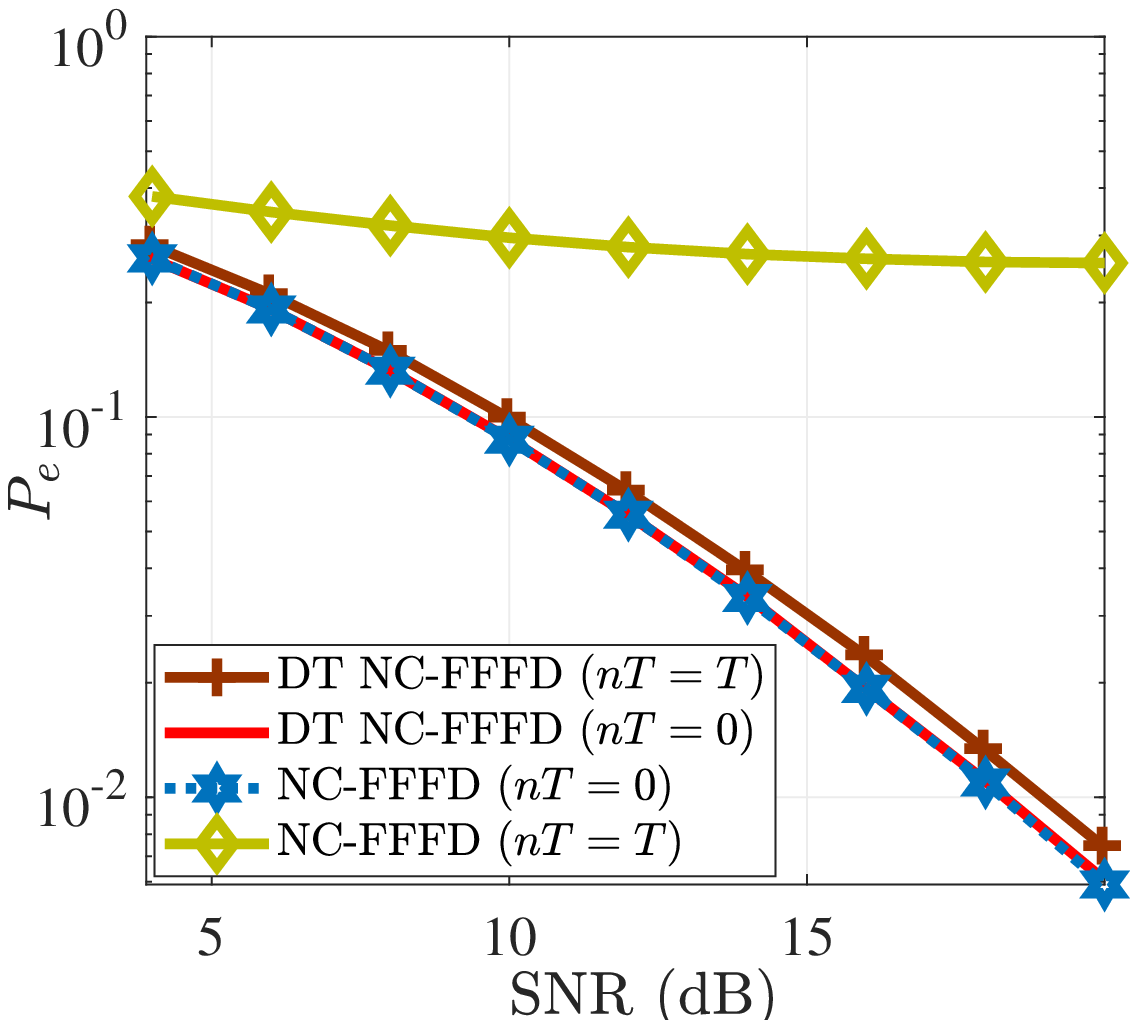}
\caption{\label{fig:DT_comp}DT NC-FFFD scheme, when $nT=0$ and $nT=T$ for $M=2$, $N_{B}=8$, $\Delta_{RE}=10^{-2}$, and $\Delta_{\text{DT}}=10^{-1}$.}
    \end{minipage} 
    \end{center}
\end{figure}

We now demonstrate the performance of DT NC-FFFD scheme. For all simulation purposes, we assume  $\Delta_{RE}=10^{-2}$, and $\Delta_{\text{DT}}=10^{-1}$, in addition to simulation parameters considered in the previous sections. Further, the effective SNR at Bob, denoted by SNR\textsubscript{eff}, is given by SNR\textsubscript{eff} (dB) = $\text{SNR (dB)}-\log\left(1+\Delta_{\text{DT}}\right)$. In Fig.~\ref{fig:DT_M2} and Fig.~\ref{fig:DT_M4}, we plot the error performance of DT NC-FFFD scheme as a function of SNR\textsubscript{eff} for $M=2$ and $M=4$, respectively, when $N_{B}=2,4,8,16$. From these plots, we show that the error performance of DT NC-FFFD improves as a function of SNR\textsubscript{eff}. However, to achieve this performance Charlie must use more receive-antennas as compared to its NC-FFFD counterpart. In Fig.~\ref{fig:opt_ant_DT}, we plot the optimal receive-antennas at Charlie, denoted by $N_{C}^{\ddagger}$, as a function of SNR for various combinations of $M$ and $N_{B}$, and observe that since $\alpha$ is a function of $N_{o}$, the number of receive-antennas required by Charlie is an increasing function of SNR. Further, it is clear from the plot that we need to mount more receive-antennas at Charlie for DT NC-FFFD scheme as compared to NC-FFFD scheme. Furthermore, we also plot the error performances of NC-FFFD and DT NC-FFFD schemes in Fig~\ref{fig:DT_comp}, for the case when $nT=0$ and $nT=T$, when $M=2$ and $N_{B}=8$. From the plots, we find that, when $nT=0$, the error performance of NC-FFFD and DT NC-FFFD exactly overlaps. However, when $nT=T$, the error-rates of DT NC-FFFD are better than the error-rates of NC-FFFD scheme. We also notice a marginal degradation in the performance of DT NC-FFFD when $nT=T$ compared to $nT=0$ due to lower effective SINR in the former case. 
  
\section{Covertness Analysis of NC-FFFD Relaying Scheme}
\label{sec:Covert}
When communicating in the presence of a reactive jamming adversary, it becomes imperative that the communication is covert. In the context of this work, covertness is the ability of Alice and Charlie to communicate without getting detected by Dave's ED or CD. Henceforth, we discuss Dave's capability to detect the proposed  countermeasures by focusing on the communication over $f_{AB}$. 

\subsection{Energy Detector (ED)}
After executing the jamming attack, Dave collects a frame of $L$ symbols on $f_{AB}$ and computes their average energy. A countermeasure is detected when the difference between the computed average energy (after the jamming attack) and the average energy (before the jamming attack) is greater than the tolerance limit $\tau$, where $\tau\geq 0$ is a small number of Dave's choice.

When no countermeasure is implemented, Dave receives symbols from Alice on $f_{AB}$. Since Dave has single receive-antenna, the $l^{th}$ symbol received by Dave on $f_{AB}$ is, $r_{D}(l) = h_{AD}(l)x(l) + n_{D}(l),\ l = 1,\cdots , L$, where, $h_{AD}(l)\sim{\cal CN}\left(0,1\right)$ is the fading channel on the $l^{th}$ symbol on Alice-to-Dave link, $n_{D}(l)\sim{\cal CN}\left(0, \tilde{N}_{o}\right)$ is the effective AWGN at Dave, such that $\tilde{N}_{o}=N_{o}+\sigma_{DD}^{2}$, where $\sigma_{DD}^{2}$ is the variance of the residual SI at Dave and $N_{o}$ is the variance of the AWGN at Dave. Further, the scalar $x(l)\in\{0,1\}$ is the $l^{th}$ symbol transmitted by Alice. Due to uncoded communication over fast-fading channel, $r_{D}(l)$ is statistically independent over $l$. The average energy received by Dave on $f_{AB}$ corresponding to $r_{D}(l)$, $l\in\{1,\cdots,L\}$ is given by, $E_{D,f_{AB}}$, where $E_{D,f_{AB}} = \frac{1}{L}\sum_{l=1}^{L}\left\vert r_{D}(l)\right\vert^{2}$. Since $h_{AD}(l)$ and the AWGN $n_{D}(l)$ are Random Variables (RV), $E_{D,f_{AB}}$ is also a RV. Using weak law of large numbers, $\frac{1}{L}\sum_{l=1}^{L}\left\vert r_{D}(l)\right\vert^{2}\rightarrow E_{f_{AB}}$ in probability, where, $E_{f_{AB}} = \tilde{N}_{o} + 0.5$ denotes the expected energy of $r_{D}(l)$ on $f_{AB}$, before the jamming attack. Since low-latency messages typically have short packet-length, Dave cannot collect a large number of observation samples. Therefore, $L$ is generally small, and with probability $1$, $E_{D,f_{AB}}\neq E_{f_{AB}}$. If $\mathcal{H}_{0}$ and $\mathcal{H}_{1}$ denote the hypothesis of no countermeasure and countermeasure, respectively, then, given $\mathcal{H}_{0}$ is true, false-alarm is an event when $E_{D,f_{AB}}$ deviates from $E_{f_{AB}}$ by an amount greater than $\tau$. We now formally define the probability of false-alarm.

\begin{definition}\label{def:pfa}
The probability of false-alarm denoted by, $\mathbf{P}_{FA}$ is given as, $\Pr\left(\left.\left\vert E_{D,f_{AB}}- E_{f_{AB}}\right\vert > \tau\right\vert \mathcal{H}_{0}\right)$, for $\tau>0$.
\end{definition}
\noindent If $u_{l}$ denotes the energy of $l^{th}$ symbol on $f_{AB}$ without any countermeasure, then the RV corresponding to the average energy of $L$ symbols is denoted by, $\mathcal{U}_{L} = \frac{1}{L}\sum_{l=1}^{L}u_{l}$. In order to compute $\mathbf{P}_{FA}$, first we compute the distribution of $\mathcal{U}_{L}$ in the next theorem.
\begin{theorem}\label{th:pdf_Um}
Given $\mathcal{H}_{0}$ is true, if $\tilde{N}_{o}<<1$, then the  PDF of $~\mathcal{U}_{L}$, i.e., $p_{\mathcal{U}_{L}}(\varsigma)$ is $\left(\frac{1}{2}\right)^{L}\sum_{l=0}^{L}{L \choose l}\frac{L^{l} e^{-L\varsigma} \varsigma^{l-1}}{\Gamma(l)}$, $\varsigma>0$. \cite[Theorem 5]{my_TCCN}
\end{theorem}
From Definition~\ref{def:pfa}, $\mathbf{P}_{FA} = \Pr\left(E_{D,f_{AB}}>E_{f_{AB}} + \tau\right) + \Pr\left(E_{D,f_{AB}}\leq E_{f_{AB}} -\tau\right)$. Therefore, using the PDF of $\mathcal{U}_{L}$ from Theorem~\ref{th:pdf_Um}, the closed-form expression of $\mathbf{P}_{FA}$ is given by,
\bieee
\mathbf{P}_{FA} &=& \dfrac{1}{2^{L}}\left(\sum_{l=0}^{L}{L \choose l}\dfrac{\Gamma\left(l, L(E_{f_{AB}}+\tau)\right)}{\Gamma(l)} + \sum_{l=0}^{L}{L \choose l}\dfrac{\gamma\left(l, L(E_{f_{AB}}-\tau)\right)}{\Gamma(l)}\right).\label{eq:pfa}
\eieee
When using NC-FFFD relaying scheme, at the $l^{th}$ symbol instant, Alice and Charlie synchronously transmit dummy OOK bit $b(l)\in\{0,1\}$ with energies $\alpha$ and $1-\alpha$, respectively, on $f_{AB}$, where $b(l)$ is the least significant bit of the pre-shared Gold sequence. The baseband symbol received at Dave is, $r_{D}(l) = h_{AD}(l)\sqrt{\alpha}b(l) + h_{CD}(l)\sqrt{1-\alpha}b(l) + n_{D}(l)$, where, for $l^{th}$ symbol, $h_{AD}(l)\sim{\cal CN}\left(0, 1\right)$ and $h_{CD}(l)\sim{\cal CN}\left(0, (1+\partial)\right)$ are Rayleigh fading channels for Alice-to-Dave and Charlie-to-Dave links, respectively. Since the location of Dave can be arbitrary, the variances of Alice-to-Dave and Charlie-to-Dave links are not identical. Thus, $\partial$ captures the relative difference in the variance. If $E_{D,f_{AB}}^{FF}$ denotes the average energy received at Dave, when Alice and Charlie use NC-FFFD scheme, then due to change in the signal model, $E_{D,f_{AB}}^{FF}\neq E_{D,f_{AB}}$. Along the similar lines of $\mathbf{P}_{FA}$, we now formally define the probability of detection at Dave when NC-FFFD scheme is implemented.
\begin{definition}\label{def:pd}
If $\mathbf{P}_{D}$ denotes the probability of detection at Dave when $\mathcal{H}_{1}$ true, then for any $\tau>0$, $\mathbf{P}_{D} = \Pr\left(\left.\left\vert E_{D,f_{AB}}^{FF}-E_{f_{AB}}\right\vert > \tau\right\vert \mathcal{H}_{1}\right)$.
\end{definition}
Further, if $v_{l}$ denotes the energy of $l^{th}$ received symbol when using the countermeasure, then $\mathcal{V}_{L}$ denotes the average energy of $L$ symbols, where, $\mathcal{V}_{L} = \frac{1}{L}\sum_{l=1}^{L}v_{l}$. We provide the closed-form expression of PDF of $v_{l}$ and $\mathcal{V}_{L}$ in the next theorem.
\begin{theorem}
\label{th:pdf_Vm}
When $\tilde{N}_{o}<<1$ and $\mathcal{H}_{1}$ is true, the PDF of $\mathcal{V}_{L}$, i.e., $p_{\mathcal{V}_{L}}(\varsigma)$, is given by, 
\bieee
\left(\frac{1}{2}\right)^{L}\sum_{l=0}^{L}{L \choose l}\frac{\left(\frac{L}{\mathcal{A}}\right)^{l} e^{-\frac{L}{\mathcal{A}}\varsigma} \varsigma^{l-1}}{\Gamma(l)}, \text{ where, $\varsigma>0$ and $\mathcal{A} = \alpha + (1-\alpha)(1+\partial)$.}\label{eq:pfd_Vm}
\eieee
\end{theorem}
\begin{proof}
When $\mathcal{H}_{1}$ is true, the received symbol at Dave on $f_{AB}$ is given as, $r_{D}(l) = h_{AD}(l)\sqrt{\alpha}b(l) + h_{CD}(l)\sqrt{1-\alpha}b(l) + n_{D}(l)$. Thus, the PDF of $v_{l}$ is given as, $p_{v_{l}}\left(\varsigma\right) = \frac{1}{2}\left(\dfrac{1}{\tilde{N}_{o}}e^{-\frac{\varsigma}{\tilde{N}_{o}}} + \frac{1}{\tilde{N}_{o}+\mathcal{A}}e^{-\frac{\varsigma}{\tilde{N}_{o}+\mathcal{A}}}\right)$, where $\mathcal{A} = \alpha + (1-\alpha)(1+\partial)$. Computing $\mathcal{V}_{M}$ requires us to sum  $L$ independent $v_{l}$ random variables each scaled by $L$, i.e., sum of $L$ independent ${v_{m}}/{M}$ random variables. Therefore,
\bieee
p_{v_{l}/L}(\varsigma) = \dfrac{1}{2}\left(\dfrac{1}{\tilde{N}_{o}}e^{-\frac{L\varsigma}{\tilde{N}_{o}}} + \dfrac{L}{\tilde{N}_{o}+\mathcal{A}}e^{-\frac{L\varsigma}{\tilde{N}_{o}+\mathcal{A}}}\right)\approx \dfrac{1}{2}\left(\delta(\varsigma) + \dfrac{L}{\mathcal{A}}e^{-\frac{L}{\mathcal{A}}\varsigma}\right),\nn
\eieee
\noindent where, the approximation is because $\tilde{N}_{o}<<1$. The pdf of $\mathcal{V}_{L}$ is equivalent to computing $L$-fold convolution ($*L$) of $p_{v_{l}/L}(\varsigma)$. For simplification, we use the properties of Laplace transform ($\mathscr{L}[\cdot]$) and inverse Laplace transform ($\mathscr{L}^{-1}[\cdot]$) to compute the pdf of $\mathcal{V}_{L}$ as,  
\bieee
p_{\mathcal{V}_{L}}(\varsigma) = \left(p_{v_{l}/L}(\varsigma)\right)^{*L} &=& \mathscr{L}^{-1}\left[\left(\mathscr{L}\left[p_{v_{l}/L}(\varsigma)\right]\right)^{L}\right] = \mathscr{L}^{-1}\left[\left(\dfrac{1}{2}\right)^{L}\left[1 + \frac{L/\mathcal{A}}{\left(s+L/\mathcal{A}\right)}\right]^{L}\right].\nn
\eieee
Using binomial expansion we expand $\left[1 + \frac{L/\mathcal{A}}{\left(s+L/\mathcal{A}\right)}\right]^{L}$ and substitute in above to obtain \eqref{eq:pfd_Vm}.
\end{proof}

Overall, from Definition~\ref{def:pd} we have, $\mathbf{P}_{D} = \Pr\left(E_{D,f_{AB}}^{FF}>E_{f_{AB}} + \tau\right) + \Pr\left(E_{D,f_{AB}}^{FF}\leq E_{f_{AB}} -\tau\right)$. Thus, the probability of miss-detection, $\mathbf{P}_{MD}$ is given by $1-\mathbf{P}_{D}$. From the result of Theorem~\ref{th:pdf_Vm},
\bieee
 \mathbf{P}_{MD} &=& 1 -\dfrac{1}{2^{L}}\sum_{l=0}^{L}{L \choose l}\dfrac{\Gamma\left(l, \frac{L}{\mathcal{A}}(E_{f_{AB}}+\tau)\right)}{\Gamma(l)} - \dfrac{1}{2^{L}}\sum_{l=0}^{L}{L \choose l}\dfrac{\gamma\left(l, \frac{L}{\mathcal{A}}(E_{f_{AB}}-\tau)\right)}{\Gamma(l)}.\label{eq:pmd}
\eieee
Ideally, a low $\mathbf{P}_{FA}$ and a high $\mathbf{P}_{D}$ allows Dave to detect a countermeasure. However, the legitimate nodes would like to drive the sum $\mathbf{P}_{FA}+\mathbf{P}_{MD}$ close to $1$ for any value of $\tau$.
\begin{rem}
If $\partial=0$ in \eqref{eq:pmd}, then $\mathbf{P}_{FA} + \mathbf{P}_{MD} = 1$, for all $\alpha\in (0,1)$ and arbitrary $M$ and $\tau$.
\end{rem}
\begin{figure}[!htb]
    \centering
    \begin{subfigure}[b]{0.48\textwidth}
       \centering
        \includegraphics[width = 0.8\textwidth, height = 5cm]{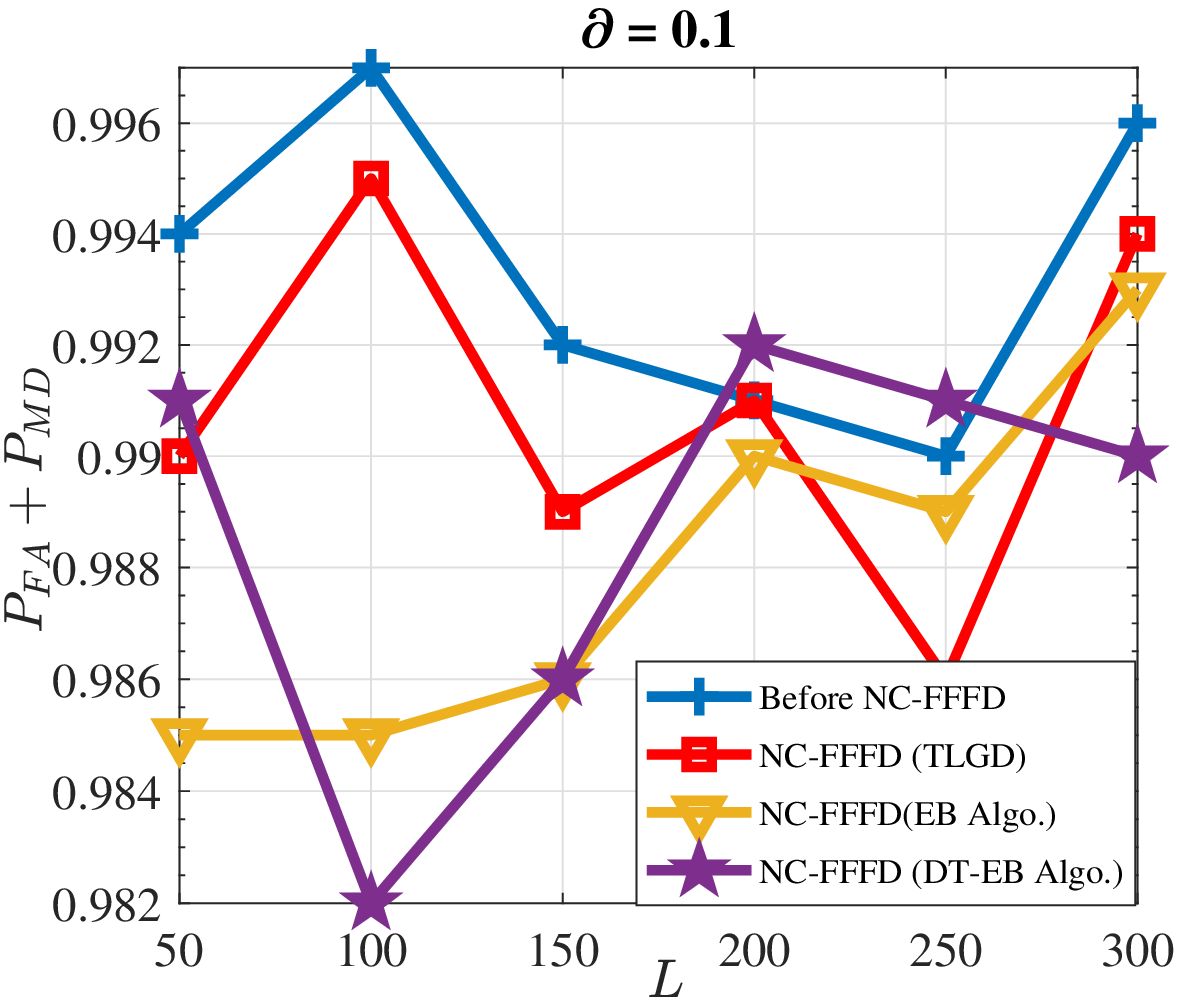}
     \end{subfigure}%
        \hfill
    \begin{subfigure}[b]{0.48\textwidth}
        \centering
        \includegraphics[width = 0.8\textwidth, height = 5cm]{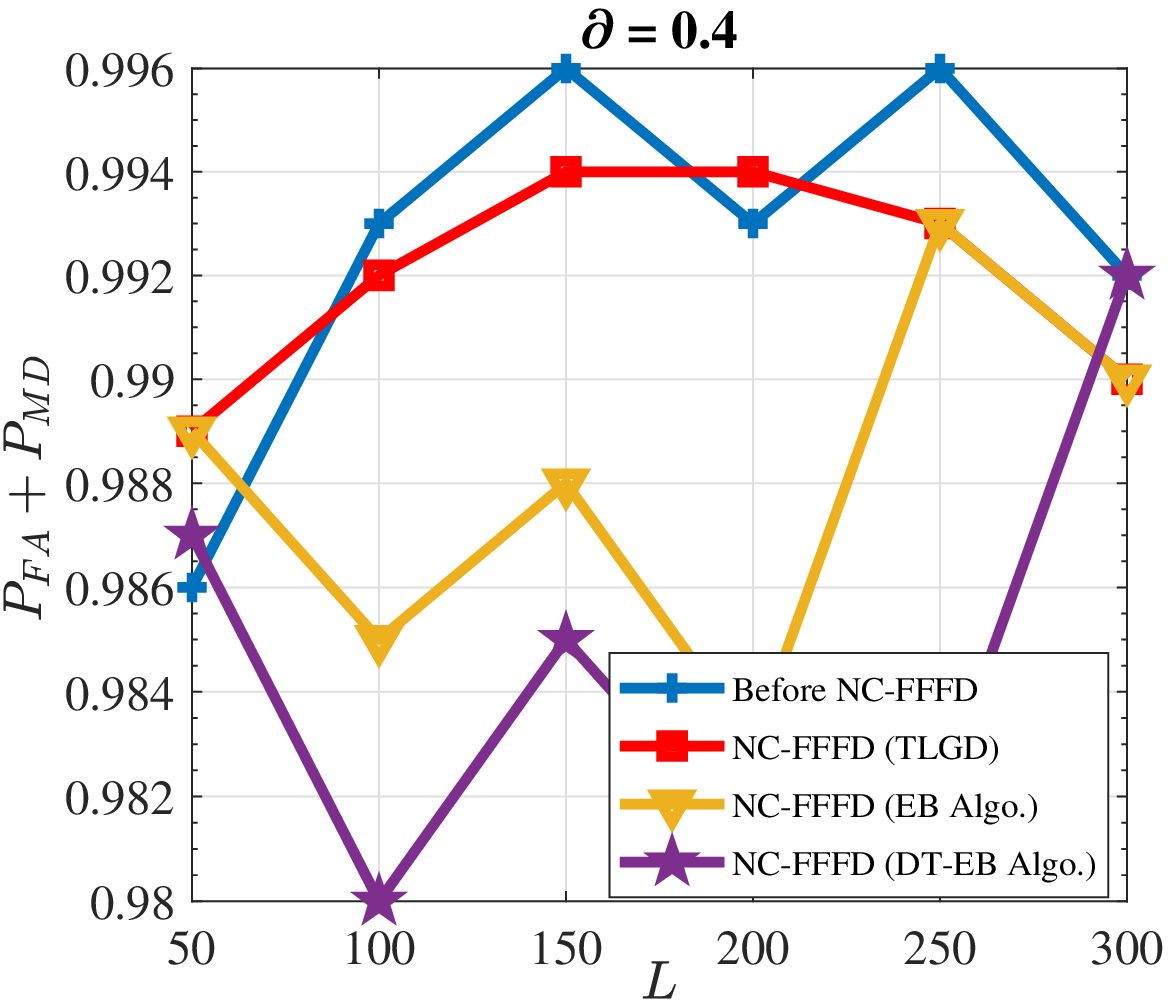}
    \end{subfigure}
    \caption{\label{fig:PFA_PMD} $\mathbf{P}_{FA} + \mathbf{P}_{MD}$ as a function of $L$ and $\partial$ at 25 dB (including the residual SI), $N_{B}=8$, and $\Delta_{\text{DT}}=0.1$.}
\end{figure}
Although, the above result theoretically guarantees $\mathbf{P}_{FA} + \mathbf{P}_{MD} = 1$ for specific value of $\partial$, Fig.~\ref{fig:PFA_PMD} shows the simulation results for  $\partial=0.1$ and $\partial=0.4$. Since Alice and Charlie transmit $b(l)$, with energy $\alpha$ and $1-\alpha$, respectively, the communication on $f_{AB}$ is independent of $\{\epsilon_{j},\eta_{j}\vert j=1,\cdots,M\}$ and only depends on the value of $\alpha$. In the previous sections, we have computed the values of $\alpha$ using TLGD algorithm, EB algorithm, and DT-EB algorithm. Therefore, in Fig.~\ref{fig:PFA_PMD}, we plot the sum $\mathbf{P}_{FA} + \mathbf{P}_{MD}$ at Dave as a function of $L$, when the optimal value of $\alpha$ is computed using TLGD, EB, and DT-EB algorithm. For $\text{SNR}=25$ dB, $N_{B}=8$, and $\Delta_{\text{DT}}=0.1$, we observe that the sum $\mathbf{P}_{FA} + \mathbf{P}_{MD}\approx 1$, despite a large number of samples at Dave. These simulation results indicate that the ED at Dave is oblivious to the countermeasures implemented by Alice and Charlie.

\subsection{Correlation Detector (CD)}
In order to prevent Alice and Charlie from using repetitive coding across the frequencies \cite{my_PIMRC} and \cite{my_TCCN}, Dave deploys a CD to capture the correlation between the symbols on the jammed frequency and other frequencies in the network. Amongst  several methods to estimate the correlation, Dave uses a CD that estimates the correlation in terms of Mutual Information (MI) to capture both, linear as well as non-linear correlation between the samples. However, estimating MI requires estimating the underlying marginal and joint PDFs, which is hard in general. Therefore, Dave needs a non-parametric method of MI estimation that does not require him to know the underlying joint and marginal PDFs. KSG estimators \cite{KSG} based on \textbf{k}-nearest neighbours (\textbf{k}-NN) are well known for non-parametric MI estimation for their ease of implementation. Thus, Dave uses a KSG estimator to detect the proposed countermeasures. Since the information symbols on the frequency bands other than $f_{CB}$ are implicitly independent of symbols on $f_{AB}$, we only focus on estimating the correlation between the symbols on $f_{AB}$ and $f_{CB}$ as $f_{CB}$ is the helper's frequency band. In the context of this work, Dave estimates the MI between the energies of the samples on $f_{AB}$ and $f_{CB}$.

We will first show the effect of transmitting dummy OOK bit $b\in\{0,1\}$ on $f_{AB}$, from the pre-shared Gold sequence. When no countermeasure is used, and Alice and Charlie transmit independent symbols on $f_{AB}$ and $f_{CB}$, respectively, the energy scatter-plot is as shown in Fig.~\ref{fig:CD}a. If Alice and Charlie  use repetitive coding across the frequencies, the energy samples are clustered only around few centres as shown in Fig.\ref{fig:CD}b. Further, when Alice and Charlie cooperatively use Gold-sequence bits, the scatter-plot is as shown in Fig.~\ref{fig:CD}c. It can be observed that the scatter-plot in Fig~\ref{fig:CD}a and Fig~\ref{fig:CD}c are more randomised as compared to Fig.~\ref{fig:CD}b. This suggests that, when Alice and Charlie transmits bits from a Gold-sequence based scrambler on $f_{AB}$, the observations at Dave are similar to when they were transmitting independent symbols.
\begin{figure}[!hbt]
\vspace{-0.25in}
\centering
\includegraphics[width = 0.73\textwidth, height = 8.5cm]{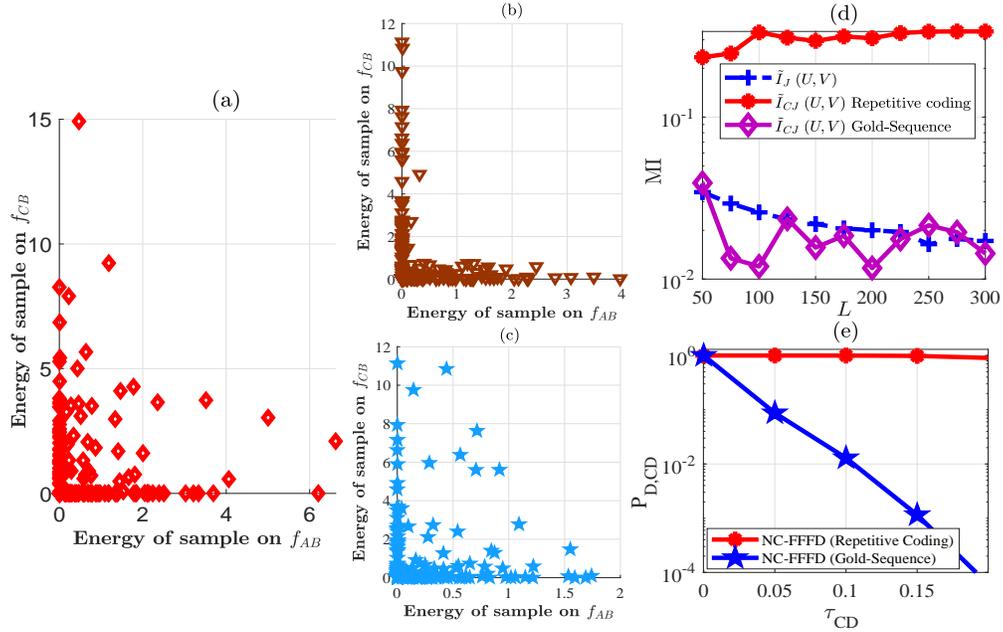}
\caption{\label{fig:CD}Scatter-plots representing the energy pairs received at Dave for SNR = 25 dB, $N_{B}=8$, $L=50$, when (a) Dave is not jamming. (b)Alice and Charlie use repetitive coding across $f_{AB}$ and $f_{CB}$. (c) Alice and Charlie cooperatively use Gold-sequence. (d) MI before jamming and after using  NC-FFFD with Gold-sequence and with Repetitive coding as a function of $L$ at SNR = 25 dB, $\mathbf{k}=2$, and $N_{B}=8$. (e) $\mathbf{P}_{\text{D,CD}}$ when NC-FFFD is implemented with repetitive coding and with Gold-sequence, for $L=150$ at 25 dB, $N_{B}=8$, and $\mathbf{k}=2$.}
\end{figure}

To formally measure the correlation, let $U$ and $V$ denote the RVs corresponding to the energy of samples on $f_{AB}$ and $f_{CB}$, respectively. Before the jamming attack, since Alice and Charlie transmit independent symbols, the MI estimate denoted by, $\tilde{I}(U,V)$, should be zero. However, due to small number of samples, i.e., $L$, $\tilde{I}(U,V)$ is a small non-zero value which approaches zero as the number of samples increases. In Fig.~\ref{fig:CD}d, for 25 dB, $N_{B}=8$, and $\mathbf{k}=2$, we plot the estimated MI at Dave, before the jamming attack,  and after implementing NC-FFFD scheme with repetitive coding and Gold-sequence. We use $10^{3}$ iterations for each value of $L$ to compute $\tilde{I}(U,V)$  before the jamming attack because Dave wants to have a good estimate of MI. However, after the jamming attack, Dave cannot use multiple iterations for estimating $\tilde{I}(U,V)$ for a given $L$. From Fig.~\ref{fig:CD}d, it is clear that the estimated MI for NC-FFFD is high and an increasing function of $L$ when repetitive coding is used. In contrast, when Gold-sequence is used, the estimated MI oscillates near the MI estimate  before the jamming attack.

If $\tau_{\text{CD}}$ denotes the resolution of Dave's CD, then probability of detection denoted by, $\mathbf{P}_{\text{D,CD}}$ is given as, $\mathbf{P}_{\text{D,CD}} = \Pr\left(\left.\left\vert\mathbb{E}\left[\tilde{I}_{\text{J}}(U,V)\right]-\tilde{I}_{\text{CJ}}(U,V)\right\vert \geq \tau_{\text{MID}}\right\vert\mathcal{H}_{1}\right)$, where, $\mathbb{E}\left[\tilde{I}_{\text{J}}(U,V)\right]$ denotes the long term estimate of MI before the jamming attack and $\tilde{I}_{\text{CJ}}(U,V)$ denotes the estimate of MI after implementing NC-FFFD with pre-shared Gold-sequence bits. In Fig.~\ref{fig:CD}e, we plot $\mathbf{P}_{\text{D,CD}}$ at SNR = 25 dB, $N_{B}=8$, and $L=150$ samples, where the optimal value of $\alpha$ is computed using the EB algorithm. Since $\tau_{\text{CD}}$ determines the accuracy of the CD, we assume $0\leq\tau_{\text{CD}}\leq 0.20$, because a very large resolution results in poor detection. When NC-FFFD is implemented and Alice uses  repetitive coding, $\mathbf{P}_{\text{D,CD}}$ is close to $1$ for the given range of $\tau_{\text{CD}}$. However, when NC-FFFD is implemented using Gold-sequence bits $b$, $\mathbf{P}_{\text{D,CD}}$ reduces as a function of $\tau_{\text{CD}}$. Thus, when Alice and Charlie use Gold-sequence bits to transmit on $f_{AB}$, the probability of detecting the countermeasure at Dave is small, as the symbols on $f_{AB}$ and $f_{CB}$ are independent by design.

\section{Conclusion} 
In this work, we have envisaged a strong FD adversary who injects jamming energy on the low-latency messages of the victim in a fast-fading environment. Unlike the reactive adversaries in the literature, the adversary in our model, uses an energy detector and a correlation detector to prevent the use of any pre-existing countermeasures. Against this threat model, we have proposed NC-FFFD relaying scheme, where the victim seeks help from a helper to fast-forward its symbols to the base station. Based on the helper's data-rate we have derived analytical results on the joint error performance, and then have proposed a family of algorithms to compute the near-optimal energy levels at the victim and the helper nodes. Further, we have also shown that, with high probability, the proposed scheme successfully engages the reactive adversary to the jammed frequency. Overall, this is the first work of its kind to address security threats from a reactive adversary in a fast-fading environment.

\ifCLASSOPTIONcaptionsoff
  \newpage
\fi
\bibliography{Ref} 
\bibliographystyle{IEEEtran}

\end{document}